\documentclass[11pt]{article}

\usepackage{ambecomen}
\usepackage{graphicx}
\usepackage{amsmath}
\usepackage{amsmath,scalefnt}
\usepackage{footnote}
\usepackage{fmtcount}
\usepackage{color}
\usepackage{threeparttable}
\usepackage{chemarr}
\newcommand{\arrowchem}[1]{\xrightarrow{#1}}
\newcommand{\arrowschem}[2]{\xrightleftharpoons[#1]{#2}}
\newcommand{\nnorm}[1]{\mathbf{\big\|} #1 \mathbf{\big\|^*}_{p,Q}}
\newcommand{\norm}[1]{\left\Vert #1 \right\Vert}
\newcommand{\normpQ}[1]{\norm{#1}_{p,Q}}
\newcommand{\MpQ}{M_{p,Q}}

\title{Logarithmic Lipschitz norms \\
and diffusion-induced instability}
\author{{Zahra Aminzare and Eduardo D. Sontag}\\ \\
{\small Department of Mathematics, Rutgers University,}\\
{\small Piscataway, NJ 08854-8019 USA}}

\begin{document}
\maketitle

\begin{abstract}
This paper proves that contractive ordinary differential equation systems
remain contractive when diffusion is added.  Thus, diffusive instabilities, in
the sense of the Turing phenomenon, cannot arise for such systems. An
important biochemical system is shown to satisfy the required conditions.
\end{abstract}

\section{Introduction} \label{introduction}

In this work, we study reaction-diffusion PDE systems
\[
u_t \;= \;F(u)+D\Delta u
\]
as well as their discrete analogues (``compartmental-systems'').
Here,
\[
u(\omega,t)\,=\,(u_1(\omega,t ),\ldots ,u_n(\omega,t )) \,,
\quad
u_t \,=\,
\left(\frac{\partial u_1}{\partial t},\ldots ,\frac{\partial u_n}{\partial t}\right)\,,
\]
$\Delta $ is the Laplacian operator on a suitable spatial domain $\Omega $,
and no flux (Neumann) boundary conditions are assumed.

In biology, a PDE system of this form describes individuals (particles,
chemical species, etc.) of $n$ different types, with respective abundances
$u_i(\omega,t )$ at time $t$ and location $\omega \in \Omega $, that can react instantaneously,
guided by the interaction rules encoded into the vector field $F$, and can
diffuse due to random motion.
Reaction-diffusion PDE's play a key role in modeling intracellular dynamics
and protein localization in cell processes such as cell division and
eukaryotic chemotaxis
(e.g., \cite{kholodenko1999,%
kalab2002,%
kholodenko2006,%
xiong2011})
as well as in the modeling of differentiation in multi-cellular organisms,
through the diffusion of morphogens which control heterogeneity in gene
expression in different cells
(e.g. \cite{murray2002,%
keshet}).
From a bioengineering perspective, reaction-diffusion models can be used to
model artificial mechanisms for achieving cellular heterogeneity in tissue
homeostasis (e.g.,
\cite{basu2005,%
weiss12homeostasis}).

The ``symmetry breaking'' phenomenon of diffusion-induced, or Turing,
instability refers to the case where a dynamic equilibrium $\bar u$ of the
non-diffusing ODE system $u_t=F(u)$ is stable, but, at least for some diagonal
positive matrices $D$, the corresponding uniform state $u(\omega)=\bar u$ is
unstable for the PDE system $u_t=F(u)+D\Delta u$.
This phenomenon has been studied at least since Turing's seminal work on
pattern formation in morphogenesis~\cite{Turing}, where he argued that
chemicals might react and diffuse so as result in heterogeneous spatial
patterns.
Subsequent work by Gierer and Meingardt~\cite{GiererMeinhardt1972,Gierer1981}
produced a molecularly plausible minimal model, using two substances that
combine local autocatalysis and long-ranging inhibition.
Since that early work, a variety of processes in physics, chemistry, biology,
and many other areas have been studied from the point of view of diffusive
instabilities, and the mathematics of the process has been extensively studied
\cite{
Othmer1969,
Segel1972,
Cross1978,
Conway1978,
Othmer1980,
murray2002,
Borckmans1995,
Castets1990,
keshet,
maini2000}.
Most past work has focused on local stability analysis, through the analysis
of the instability of nonuniform spatial modes of the linearized PDE.
Nonlinear, global, results are usually proved under strong constraints on
diffusion constants as they compare to the growth of the reaction part.

In this note, we are interested in conditions on the reaction part $F$ that
guarantee that no diffusion instability will occur, no matter what is the size
of the diffusion matrix $D$.
We show that if the reaction system is ``contractive'' in the sense that
trajectories
globally and exponentially
converge to each other with respect to a diagonally weighted
$L^p$ norm, then the same property is inherited by the PDE. In particular, if
there is an equilibrium $F(\bar u)=0$, it will follow that this equilibrium
is globally exponentially stable for the PDE system.  A similar result is
also established for a discrete analog, in which a set of ODE systems are
diffusively interconnected.
We were motivated by the desire to understand the important biological systems
described in
\cite{ddv07,Russo}
for which, as we will show, contractivity holds for diagonally weighted $L^1$
norms, but not with respect to diagonally weighted $L^p$ norms, for any
$1<p\leq\infty$.

Closely related work in the literature has dealt with the synchronization
problem, in which one is interested in the convergence of trajectories to
their space averages in weighted $L^2$ norms, for appropriate diffusion
coefficients and Laplacian eigenvalues, specifically
\cite{IEEEsysbio_JAS},
which used passivity ideas from control theory for
systems with special structures such as cyclic systems,
\cite{limingwang2012} which extended this approach to more general passive
structures, and~\cite{Arcak} which obtained a generalization involving
a contraction-like diagonal stability condition.
Our work uses very different techniques, from nonlinear functional analysis
for normed spaces, than the quadratic Lyapunov function approaches,
appropriate for Hilbert spaces, followed in these references.

\section{Logarithmic Lipschitz constants and norms}
We start by reviewing several useful concepts from nonlinear functional analysis, and proving certain technical properties for them.
\subsection{General normed spaces}
 
\begin{definition}\cite{Deimling, Soderlind} Let $(X,\|\cdot\|_X)$ be a normed space. For $x_1, x_2\in X$, the right and left semi inner products are defined by
\begin{equation}\label{semi-inner}
(x_1, x_2)_{\pm}=\displaystyle\|x_1\|_X\lim_{h\to 0^{\pm}}\frac{1}{h}\left(\|x_1+hx_2\|_X-\|x_1\|_X\right).
\end{equation}
\end{definition}

\begin{remark}\label{existence-of-norm}
As every norm possesses left and right Gateaux-differentials, the limits in $(\ref{semi-inner})$ exist and are finite. For more details see \cite{sontag-control}.
\end{remark}

\begin{remark}
The right and left semi inner products $(\cdot,\cdot)_{\pm}$, induce the norm $\|\cdot\|_X$ in the usual way: $(x,x)_{\pm}=\|x\|_X^2$. Conversely if the norm arises from an inner product $(\cdot, \cdot)$, as when $X$ is a Hilbert space, $(x_1, x_2)_+= (x_1, x_2)_-= (x_1, x_2)$. Moreover the right and left semi inner products satisfy the Cauchy-Schwarz inequalities:
\[-\|x\|\cdot\|y\|\leq(x, y)_{\pm}\leq\|x\|\cdot\|y\|.\]

\end{remark}

The following elementary properties of semi inner products are consequences of the properties of norms. See \cite{Deimling, Soderlind} for the proof.  
\begin{proposition}For $x, y, z\in X$ and $\alpha\geq 0$, 
\begin{enumerate}
\item $(x, -y)_{\pm}=-(x, y)_{\mp}$;
\item $(x, \alpha y)_{\pm}=\alpha(x, y)_{\pm}$;
\item $(x, y)_-+(x, z)_{\pm}\leq(x, y+z)_{\pm}\leq (x,y)_++(x,z)_{\pm}$.
\end{enumerate}
\end{proposition}

\begin{remark}
In general, the semi inner product is not symmetric: \[(x, y)_{\pm}\neq(y, x)_{\pm}.\]
\end{remark}

\begin{definition}\cite{Soderlind} Let $(X,\|\cdot\|_X)$ be a normed space and $f\colon Y \to X$ be a function, where $Y\subseteq X$. The strong least upper bound logarithmic Lipschitz constants of $f$ induced by the norm $\|\cdot\|_X$, on $Y$, are defined by
$$M_{Y, X}^{\pm}[f]=\displaystyle\sup_{u\neq v\in Y}\frac{(u-v,f(u)-f(v))_{\pm}}{\|u-v\|_X^2},$$

or equivalently
\begin{equation}\label{defM+}
M_{Y, X}^{\pm}[f]=\displaystyle\sup_{u\neq v\in Y}\lim_{h\to0^{\pm}}\frac{1}{h}\left(\frac{\|u-v+h(f(u)-f(v))\|_X}{\|u-v\|_X}-1\right).
\end{equation}
If $X=Y$, we write $M_X^{\pm}$ instead of $M_{X, X}^{\pm}$.
\end{definition}

\begin{proposition}\label{subadd}
 Let $(X,\|\cdot\|_X)$ be a normed space. For any $f$, $g\colon Y\to X$ and any $Y\subseteq X$:
\begin{enumerate}
\item $M_{Y, X}^+[f+g]\leq M_{Y, X}^+[f]+M_{Y, X}^+[g]$;
\item $M_{Y, X}^{\pm}[\alpha f]=\alpha M_{Y, X}^{\pm}[f]$ for $\alpha\geq0$.
\end{enumerate}
\end{proposition}

\begin{proof}
\begin{enumerate}
\item By the definition of $M_{Y, X}^{\pm}$, and the triangle inequality for norms, we have
\[
\begin{array}{lcl}
M_{Y, X}^+[f+g]&=&\displaystyle\sup_{u\neq v\in Y}\lim_{h\to0^+}\frac{1}{h}\left(\frac{\|u-v+h((f+g)(u)-(f+g)(v))\|_X}{\|u-v\|_X}-1\right)\\
&=&\displaystyle\sup_{u\neq v\in Y}\lim_{h\to0^+}\frac{1}{2h}\left(\frac{\|2(u-v)+2h((f+g)(u)-(f+g)(v))\|_X}{\|u-v\|_X}-2\right)\\
&\leq&\displaystyle\sup_{u\neq v\in Y}\lim_{h\to0^+}\frac{1}{2h}\displaystyle\left(\frac{\|u-v+2h(f(u)-f(v))\|_X}{\|u-v\|_X}-1\right)+\\
& &\displaystyle\sup_{u\neq v\in Y}\lim_{h\to0^+}\frac{1}{2h}\displaystyle\left(\frac{\|u-v+2h(g(u)-g(v))\|_X}{\|u-v\|_X}-1\right)\\
&=& M_{Y, X}^+[f]+M_{Y, X}^+[g]
\end{array}
\]
\item For $\alpha=0$, the equality is trivial, because both sides are equal to zero. For $\alpha>0$:
\[
\begin{array}{lcl}
M_{Y, X}^{\pm}[\alpha f]&=&\displaystyle\sup_{u\neq v\in Y}\lim_{h\to0^{\pm}}\frac{1}{h}\left(\frac{\|u-v+h(\alpha f(u)-\alpha f(v))\|_X}{\|u-v\|_X}-1\right)\\
&=& \displaystyle\sup_{u\neq v\in Y}\lim_{h\to0^{\pm}}\frac{\alpha}{\alpha h}\left(\frac{\|u-v+(\alpha h)(f(u)-f(v))\|_X}{\|u-v\|_X}-1\right)\\
&=&\alpha M_{Y, X}^{\pm}[f].
\end{array}
\]
\end{enumerate}
\end{proof}

\begin{definition}\cite{Soderlind} Let $(X,\|\cdot\|_X)$ be a normed space and $f\colon Y\to X$ be a function, where $Y\subseteq X$. The least upper bound Lipschitz constant of $f$ induced by the norm $\|\cdot\|_X$, on $Y$, is defined by
 $$L_{Y,X}[f]=\displaystyle\sup_{u\neq v\in Y}\frac{\|f(u)-f(v)\|_X}{\|u-v\|_X}.$$
 Note that $L_{Y,X}[f]<\infty$ if and only if $f$ is Lipschitz on $Y$.
\end{definition}

\begin{definition}
\cite{Soderlind}
Let $(X,\|\cdot\|_X)$ be a normed space and $f\colon Y \to X$ be a Lipschitz function. The least upper bound logarithmic Lipschitz constant of $f$ induced by the norm $\|\cdot\|_X$, on $Y\subseteq X$, is defined by
 $$M_{Y,X}[f]=\displaystyle\lim_{h\to0+}\frac{1}{h}\left(L_{Y,X}[I+hf]-1\right),$$
 or equivalently
\begin{equation}\label{defM}
M_{Y, X}[f]=\displaystyle\lim_{h\to0^+}\sup_{u\neq v\in Y}\frac{1}{h}\left(\frac{\|u-v+h(f(u)-f(v))\|_X}{\|u-v\|_X}-1\right)
\end{equation}
 If $X=Y$, we write $M_X$ instead of $M_{X, X}$.
\end{definition}

\begin{lemma}\label{M+<M}
 $M_{Y,X}^+[f]\leq M_{Y,X}[f]$.
 \end{lemma}
 
 \begin{proof}
 For any fixed $u\neq v\in Y$, 
 $$\displaystyle\frac{\|u-v+h(f(u)-f(v))\|_X}{\|u-v\|_X}\leq\displaystyle\sup_{u\neq v\in Y}\frac{\|u-v+h(f(u)-f(v))\|_X}{\|u-v\|_X}.$$
 Now using this inequality, we have:
 $$\displaystyle\lim_{h\to0^+}\frac{1}{h}\left(\frac{\|u-v+h(f(u)-f(v))\|_X}{\|u-v\|_X}-1\right)\leq\displaystyle\lim_{h\to0^+}\sup_{u\neq v\in Y}\frac{1}{h}\left(\frac{\|u-v+h(f(u)-f(v))\|_X}{\|u-v\|_X}-1\right).$$

 Since this inequality holds for any $u\neq v\in Y$, taking $\sup$ we have:
 $$\displaystyle\sup_{u\neq v\in Y}\lim_{h\to0^+}\frac{1}{h}\left(\frac{\|u-v+h(f(u)-f(v))\|_X}{\|u-v\|_X}-1\right)\leq\displaystyle\lim_{h\to0^+}\sup_{u\neq v\in Y}\frac{1}{h}\left(\frac{\|u-v+h(f(u)-f(v))\|_X}{\|u-v\|_X}-1\right),$$
 from which the conclusion follows using $(\ref{defM+})$ and $(\ref{defM})$.
\end{proof}

\subsection{Finite dimensional case} 
The least upper bound (lub) logarithmic Lipschitz constant generalizes the usual logarithmic norm; for every matrix $A$ we have $M_X [A] = \mu_X[A]$. For ease of references, we review next the basic properties of logarithmic norms for finite dimensional operators. 

\begin{definition}
Let $(X,\|\cdot\|_X)$ be a finite dimensional normed vector space over $\r$ or $\cp$. The space $\mathcal{L}(X, X)$ of linear transformations $A\colon X \to X$ is also a normed vector space with the induced operator norm $$\|A\|_{X\to X}=\displaystyle\sup_{\|x\|_X=1}\|Ax\|_X.$$ The logarithmic norm $\mu_X(\cdot)$ induced by $\|\cdot\|_X$ is defined as the directional derivative of the matrix norm, that is,
\[
\mu_X(A)=\displaystyle\lim_{h\to 0^+}\frac{1}{h}\left(\|I+hA\|_{X\to X}-1\right),
\]
where $I$ is the identity operator on $X$. 
\end{definition}

\begin{remark}
Since $\displaystyle\sup_{s\in S}(as+b)=a\displaystyle\sup_{s\in S}(s)+b$, whenever $a>0$ and $S\subseteq\r$, it follows that 
\[
\mu_X(A)=\displaystyle\lim_{h\to 0^+}\sup_{{\|x\|_X=1}}\frac{1}{h}\left(\|x+hAx\|_X-1\right). 
\]
\end{remark}

\begin{theorem}\label{mu} For any matrix $A$,
$\mu(A)=\displaystyle\sup_{\|v\|=1}\lim_{h\to 0^+}\frac{1}{h}(\|v+hAv\|-1)$.
\end{theorem}
See the Appendix for the proof. 

 \begin{corollary}\label{mu=M+}
  Let $(X, \|\cdot\|_X)$ be a finite dimensional normed space. For any linear operator $\mathcal{L}\colon X\to X$, 
  \[\mu_X(\mathcal{L})=M_X^+[\mathcal{L}]=M_X[\mathcal{L}].\]
   \end{corollary}
   
  \begin{proof}
  The proof is immediate from the definition of $M_X$, $M_X^+$, and Theorem \ref{mu}. 
   \end{proof}

When $X=\r^n$ or $\mathbb{C}^n$, we identify operators and their matrix representations on the standard basis, and we call the logarithmic norm the {\em matrix measure}. In Table $\ref{tab-mu}$, the algebraic expression of the logarithmic norms induced by the $L^p$ norm for $p=1,2,$ and $\infty$ are shown for matrices. For proofs, see for instance \cite{Desoer}. 

\begin{table}[htdp]
\caption{\scshape Standard matrix measures for a real $n\times n$ matrix, $A=[a_{ij}]$.}
\begin{center}
\begin{tabular}{|c|c|}
\hline
vector norm, $\|\cdot\|$ & induced matrix measure, $\mu(A)$\\
\hline
$\|x\|_1=\displaystyle\sum_{i=1}^{n}\abs{x_i}$ & $\mu_1(A)=\displaystyle\max_{j}\left(a_{jj}+\displaystyle\sum_{i\neq j}\abs{a_{ij}}\right)$\\
\hline
$\|x\|_2=\left(\displaystyle\sum_{i=1}^{n}\abs{x_i}^2\right)^{\frac{1}{2}}$ & $\mu_2(A)=\displaystyle\max_{\lambda\in\spec{A}}\left(\lambda\left\{\frac{A+A^T}{2}\right\}\right)$\\
\hline
$\|x\|_{\infty}=\displaystyle\max_{1\leq i\leq n} \abs{x_i}$ & $\mu_{\infty}(A)=\displaystyle\max_{i}\left(a_{ii}+\displaystyle\sum_{i\neq j}\abs{a_{ij}}\right)$\\
\hline
\end{tabular}
\end{center}
\label{tab-mu}
\end{table}%

For ease of reference, we summarize the main notations and definitions in Table \ref{tab-def}. 
\setlength{\tabcolsep}{2pt}
\begin{table}[htdp]
\caption{\scshape Basic Concepts}
\begin{center}
\scalebox{0.7}{
\begin{tabular}{|c|c|c|c|c|}
\hline
Notation & Definition & Equivalent Definition & Equivalent Definition\\
\hline
$\mu_X(A)$ & $\displaystyle\lim_{h\to 0^+}\frac{1}{h}\left(\|I+hA\|_{X\to X}-1\right)$ & $\displaystyle\lim_{h\to 0^+}\sup_{{\|x\|_X=1}}\frac{1}{h}\left(\|x+hAx\|_X-1\right)$ & $\displaystyle\sup_{{\|x\|_X=1}}\lim_{h\to 0^+}\frac{1}{h}\left(\|x+hAx\|_X-1\right)$\\
\hline
$L_{Y,X}[f]$ & $\displaystyle\sup_{u\neq v\in Y}\frac{\|f(u)-f(v)\|_X}{\|u-v\|_X}$&&\\
\hline
$M_{Y,X}[f]$ & $\displaystyle\lim_{h\to0+}\frac{1}{h}\left(L_{Y,X}[I+hf]-1\right)$&$\displaystyle\lim_{h\to0^+}\sup_{u\neq v\in Y}\frac{1}{h}\left(\frac{\|u-v+h(f(u)-f(v))\|_X}{\|u-v\|_X}-1\right)$&\\
\hline
$(x,y)_{+}$ & $\displaystyle\|x\|_X\lim_{h\to 0^{+}}\frac{1}{h}\left(\|x+hy\|_X-\|x\|_X\right)$&&\\
\hline
$M_{Y,X}^{+}[f]$ & $\displaystyle\sup_{u\neq v\in Y}\frac{(u-v,f(u)-f(v))_{+}}{\|u-v\|_X^2}$&&$\displaystyle\sup_{u\neq v\in Y}\lim_{h\to0^+}\frac{1}{h}\left(\frac{\|u-v+h(f(u)-f(v))\|_X}{\|u-v\|_X}-1\right)$\\
\hline
\end{tabular}}
\label{tab-def}
\end{center}
\end{table}

\section{Weighted $L^p$ norms}
Suppose $\Omega$, a bounded domain in $\r^m$ with smooth boundary $\partial\Omega$ and outward normal $\mathbf{n}$, and a subset $V\subseteq \r^n$ have been fixed. We denote 
\[
\mathbf{Y}=\displaystyle \left\{v\colon\bar\Omega\to V\mid\;v=(v_1,\cdots, v_n),\quad v_i\in C^2_{\r}\left(\bar\Omega\right),\quad\frac{\partial v_i}{\partial\mathbf{n}}(\xi)=0,\; \forall\xi\in\partial\Omega\;\;\forall i\right\};
\]
where $C^2_{\r}\left(\bar\Omega\right)$ is the set of twice continuously differentiable functions $\bar\Omega\to \r$.
In addition, we denote $\mathbf{X}=C_{\r^n}\left(\bar\Omega\right)$, where $C_{\r^n}\left(\bar\Omega\right)$ is the set of all continuous functions $\bar\Omega\to \r^n$. 

Note that for each $i$, for $1\leq p<\infty$, $\|v_i\|_p=\left(\displaystyle\int_{\Omega}\abs{v_i(\omega)}^p\;d\omega\right)^{\frac{1}{p}}$ and for $p=\infty$, $\|v_i\|_p=\displaystyle\sup_{\omega\in\bar\Omega}\abs{q_iv_i(\omega)}$ and both are finite because $v_i$ is a continuous function on $\bar\Omega$ and $\bar\Omega$ is a compact subset of $\r^m$.

For any $1\leq p\leq\infty$, and any nonsingular, diagonal matrix $Q=\diag{(q_1,\cdots,q_n)}$, we introduce a $Q$-weighted norm on $\mathbf{X}$ as follows:
 \begin{equation}\label{normU}
\normpQ{v}:=\left\|Q\left(\|v_1\|_p, \cdots, \|v_n\|_p\right)^T\right\|_p.
\end{equation}
Since \[\normpQ{v}=\left\{\begin{array}{ccc}\displaystyle\left(\sum_{i}\abs{q_i}^p\|{v_i}\|_p^p\right)^\frac{1}{p} & & 1\leq p<\infty \\
\displaystyle\sup_{i}\abs{q_i}\|{v_i}\|_p &  & p=\infty\end{array}\right.\] without loss of generality we will assume $q_i>0$ for each $i$. 

With a slight abuse of notation, we use the same symbol for a norm in $\r^n$: \[\|x\|_{p,Q}:=\|Qx\|_p.\]
\begin{lemma}\label{U=S}
For any $v\in \mathbf{X}$, $\normpQ{v}=\nnorm{v}$, where
\begin{equation}\label{norm}
 \nnorm{v}=\left\{\begin{array}{ccc}\left(\displaystyle\int_{\Omega}\|Qv(\omega)\|_p^p\;d\omega\right)^{\frac{1}{p}} &  & 1\leq p<\infty \\ 
 \displaystyle\sup_{\omega}\|{Qv(\omega)}\|_{\infty}&  & p=\infty\end{array}\right.
 \end{equation}
 Note that $\|Qv(\omega)\|_p^p=\displaystyle\sum_{i=1}^n\abs{q_iv_i(\omega)}^p$ and $\|{Qv(\omega)}\|_{\infty}=\displaystyle\sup_{i}\abs{q_iv_i(\omega)}$. 
 \end{lemma}
 
  \begin{proof}
  Let $Q=\diag(q_1, \cdots, q_n)$, $q_i>0$. For $1\leq p<\infty$ (the proof is analogous when $p=\infty$), by the definitions of $\normpQ{\cdot}$ and $\nnorm{\cdot}$
\[
 \begin{array}{lcl}
  \nnorm{v}&=&\displaystyle\left(\int_{\Omega}\|Qv(\omega)\|_p^p\;d\omega\right)^{\frac{1}{p}}\\
   &=&\displaystyle\left(\int_{\Omega}\|(q_1v_1(\omega), \cdots, q_nv_n(\omega))^T\|_p^p\;d\omega\right)^{\frac{1}{p}}\\
   &=&\displaystyle\left(\int_{\Omega}\abs{q_1v_1(\omega)}^p+\cdots+\abs{q_nv_n(\omega)}^p\;d\omega\right)^{\frac{1}{p}}\\
      &=&\displaystyle\left(\|{q_1v_1}\|_p^p+\cdots+\|{q_nv_n}\|_p^p\right)^{\frac{1}{p}}\\
      &=&\displaystyle\left\|\left(q_1\|v_1\|_p, \cdots, q_n\|v_n\|_p\right)^T\right\|_p\\
      &=&\displaystyle\|Q(\|v_1\|_p, \cdots, \|v_n\|_p)^T\|_p\\
      &=&\displaystyle\normpQ{v}.
 \end{array}
   \]
  \end{proof}
  
  Note that this equality between weighted $p$ norms of functions and of vectors 
depends on our having taken the matrix $Q$ to be diagonal. This is the key place where the assumption that $Q$ is diagonal is being used. 
 
\section{Main Result}
In this section, we study the reaction-diffusion PDE:
\begin{equation} \label{re-di}
\displaystyle\frac{\partial u}{\partial t}(\omega, t)=F(u(\omega, t))+D\Delta u(\omega, t)
\end{equation}
subject to the Neumann boundary condition:
\begin{equation} \label{i-c}
\frac{\partial u}{\partial\mathbf{n}}(\xi,t)=0\quad \forall\xi\in\partial\Omega,\;\;\forall t\in[0,\infty).
\end{equation} 

\begin{assumption}\label{as-pde}
In $(\ref{re-di})-(\ref{i-c})$ we assume:
\begin{itemize}
\item $F\colon V\to\r^n$ is a (globally) Lipschitz and twice continuously differentiable vector field with components $F_i$:
$$F(x)=(F_1(x), \cdots, F_n(x))^T$$ for some functions $F_i\colon V\to\r$, where $V$ is a convex subset of $\r^n$.
\item $D=\diag(d_1, \cdots, d_n)$, with $d_i>0$, is called the diffusion matrix. 
\item $\Omega$ is a bounded domain in $\r^m$ with smooth boundary $\partial\Omega$ and outward normal $\mathbf{n}$. 
\end{itemize}
\end{assumption}

\begin{definition}
By a solution of the PDE
\begin{equation*}
\displaystyle\frac{\partial u}{\partial t}(\omega, t)=F(u(\omega, t))+D\Delta u(\omega, t)
\end{equation*}
\begin{equation*} 
\frac{\partial u}{\partial\mathbf{n}}(\xi,t)=0\quad \forall\xi\in\partial\Omega,\;\;\forall t\in[0,\infty),
\end{equation*}
on an interval $[0, T)$, where $0<T\leq\infty$, we mean a function $u=(u_1, \cdots, u_n)^T$, with $u\colon \displaystyle\bar\Omega\times [0,T)\to V$, such that:
\begin{enumerate}
 \item for each $\omega\in\bar\Omega$, $u(\omega, \cdot)$ is continuously differentiable;
 \item for each $t\in[0,T)$, $u(\cdot, t)$ is in $\mathbf{Y}$; and
 \item for each $\omega\in\bar\Omega$, and each $t\in [0,T)$, $u$ satisfies the above PDE.
\end{enumerate}
\end{definition}

Theorems on existence and uniqueness for PDE's such as $(\ref{re-di})-(\ref{i-c})$ can be found in standard references, e.g. \cite{Smith, Cantrell}. One must impose appropriate conditions on the vector field, on the boundary of $V$, to insure invariance of $V$. Convexity of $V$ insures that the Laplacian also preserves $V$. Since we are interested here in estimates relating pairs of solutions, we will not deal with existence and well-posedness. Our results will refer to solutions already assumed to exist. 

Pick any $0<T\leq\infty$ and suppose that $u$ is a solution of $(\ref{re-di})-(\ref{i-c})$ defined on $\displaystyle\bar\Omega\times [0,T)$. Define $\hat u\colon[0,T)\rightarrow \mathbf{Y}$ by $\hat u(t)(\omega)=u(\omega ,t)$. Also define the function $\tilde{F}\colon \mathbf{Y}\to \mathbf{X}$ as follows: for any $u\in \mathbf{Y}$,
\[\tilde{F}(u)(\omega)\;=\;F(u(\omega))\;\;\mbox{for each}\;\;\omega\in\bar\Omega.\]
Let $A_{p,Q}\colon \mathbf{Y}\to \mathbf{X}$ denote an $n\times n$ diagonal matrix of operators on $\mathbf{Y}$ with the operators $d_i\Delta $ on the diagonal.

\begin{lemma}
Suppose that $u$ solves the PDE $(\ref{re-di})-(\ref{i-c})$, on an interval $[0,T)$, for some $T\in (0,\infty]$, and let
\[
v(\omega ,t) \;:=\; \frac{\partial u}{\partial t}(\omega ,t) \,
\]
for each $t\geq0$ and $\omega\in\bar\Omega$. We introduce $\hat{v}\colon[0,T)\to \mathbf{X}$ by $\hat{v}(t)(\omega)=v(\omega,t).$
Then, $\hat v(t)$ is the derivative of $\hat u(t)$ in the space $(\mathbf{X},\normpQ{\cdot })$,
that is:
\[
\lim_{h\rightarrow 0} \normpQ{\frac{1}{h}\left[\hat u(t+h) - \hat u(t)\right] - \hat v(t)}
\;=\; 0,
\]
for all $t\in [0,T)$.
Moreover,
\begin{equation}\label{pde-ode}
\hat v(t) = \tilde F(\hat u(t)) + A_{p,Q} (\hat u(t)).
\end{equation}
\end{lemma}

\begin{proof}
Fix $t\in[0, T)$ and $i\in\{1,\cdots,n\}$. Using the definition of $v$, we have:  
\[\lim_{h\rightarrow 0}\left|\frac{1}{h}\left[u_i(\omega, t+h) -u_i(\omega, t)\right] -v_i(\omega, t)\right|=0,\] for any $\omega\in\bar\Omega$. 
Hence for any $\epsilon>0$, there exists $h_{\omega}>0$ such that for any $0<h< h_{\omega}$, 
\[\left|\frac{1}{h}\left[u_i(\omega, t+h) -u_i(\omega, t)\right] -v_i(\omega, t)\right|<\frac{\epsilon}{2}.\]
Now since $u_i$ is a continuous function of $\omega$, there exists a ball $B_{\omega}$ centered at $\omega$ such that for any $0<h< h_{\omega}$, 
\[\left|\frac{1}{h}\left[u_i(\tilde\omega, t+h) -u_i(\tilde\omega, t)\right] -v_i(\tilde\omega, t)\right|<\epsilon\]
for all $\tilde\omega\in B_{\omega}$.
Since $\{B_{\omega}:\omega\in\bar\Omega\}$ is an open cover of $\bar\Omega$ and $\bar\Omega$ is a compact subset of $\r^m$, finitely many of these balls, namely $B_{\omega_1}, \cdots, B_{\omega_k}$, cover $\bar\Omega$. 
Now let $h_0=\min{(h_{\omega_1}, \cdots, h_{\omega_k})}$. Then, for any $0<h< h_0$ and any $\omega\in\bar\Omega$, we have
\[\left|\frac{1}{h}\left[u_i(\omega, t+h) -u_i(\omega, t)\right] -v_i(\omega, t)\right|<\epsilon.\]
Raising to the $p$-th power and taking the integral over $\Omega$ of the above inequality, we get 
  \[\displaystyle\int_{\Omega}\left|\frac{1}{h}\left[u_i(\omega, t+h) -u_i(\omega, t)\right] -v_i(\omega, t)\right|^p\;d\omega<\abs{\Omega}\epsilon^p,\]
which by the definition of $\normpQ{\cdot }$, it implies that for any $0<h<h_0$,
\[\normpQ{\frac{1}{h}\left[u(\cdot, t+h) -u(\cdot, t)\right] -v(\cdot, t)}<c\epsilon,\]
where $c=\left(\abs{\Omega}\sum_{i=1}^nq_i^p\right)^{\frac{1}{p}}$. 
Since $\epsilon>0$ was arbitrary, we have proved that 
\[
\lim_{h\rightarrow 0}\normpQ{ \frac{1}{h}\left[\hat u(t+h) - \hat u(t)\right] - \hat v(t)}
=0.\]
For a fixed $t\in[0,T)$ and any $\omega\in\bar\Omega$:
\[\begin{array}{rcl}
\hat{v}(t)(\omega)&=&v(t,\omega)\;=\;\displaystyle\frac{\partial u}{\partial t}(\omega, t)\;=\;F(u(\omega, t))+D\Delta u(\omega, t)\\
&=&\tilde F(\hat u(t))(\omega) + A_{p,Q} (\hat u(t))(\omega),
\end{array}\]
and therefore Equation (\ref{pde-ode}) holds. 
\end{proof}

\newtheorem*{proofoflemma1}{Proof of Lemma \ref{main-result1}}
\newtheorem*{proofoftheorem0}{Proof of Theorem \ref{main-result0}}
In this section we show that $(\ref{re-di})-(\ref{i-c})$ is contracting (meaning that solutions converge exponentially to each other, as $t\to+\infty$) if $\MpQ[F]<0$, where, as defined before, 
\[\MpQ[F]=\displaystyle\lim_{h\to0^+}\sup_{x\neq y\in V}\frac{1}{h}\left(\frac{\normpQ{x-y+h(F(x)-F(y))}}{\normpQ{x-y}}-1\right).\]

Now we state the main result of this section.

\begin{theorem}\label{main-result0}
Consider the PDE $(\ref{re-di})-(\ref{i-c})$ and suppose Assumption \ref{as-pde} holds. Let $c=\MpQ[F]$ for some $1\leq p\leq\infty$, and some positive diagonal matrix $Q$. Then for every two solutions $u, v$ of the PDE $(\ref{re-di}) - (\ref{i-c})$ and all $t\in[0,T)$:
\[\normpQ{\hat{u}(t)-\hat{v}(t)}\leq e^{ct}\normpQ{\hat{u}(0)-\hat{v}(0)}.\]
\end{theorem}

\begin{remark}
In terms of the PDE $(\ref{re-di})-(\ref{i-c})$, this last estimate can be equivalently written as:
\[\normpQ{u(\cdot, t)-v( \cdot, t)}\leq e^{ct}\normpQ{u(\cdot, 0)-v(\cdot, 0)}.\]
\end{remark}
Before proving the theorem, we prove a few technical lemmas.

\begin{definition}
The upper left and right Dini derivatives for any continuous function, $\Psi\colon[0, \infty)\to\r$, are defined by
\[
\begin{array}{rcl}
\left(D^{\pm}\Psi\right)(t)=\displaystyle\limsup_{h\to 0^{\pm}}\frac{1}{h}\left(\Psi(t+h)-\Psi(t)\right)
 \end{array}.
\]
Note that $D^+\Psi$ and/or $D^-\Psi$ might be infinite.
\end{definition}

 \begin{lemma}\label{Dini}
Let $\left(X,\|\cdot\|_X\right)=\left(C_{\r^n}\left(\bar\Omega\right), \normpQ{\cdot}\right)$. Let $G\colon Y \to X$ be a (globally) Lipschitz function, where $Y\subseteq X$. Let $u, v\colon [0, \infty)\to Y$ be two solutions of $\displaystyle\frac{du(t)}{dt}=G(u(t))$. Then for all $t\in [0, \infty)$, 
\begin{equation}\label{Dini-s-i}
D^+\|(u-v)(t)\|_X= \displaystyle\frac{((u-v)(t), G(u(t))-G(v(t)))_+}{\|(u-v)(t)\|_X^2}\|(u-v)(t)\|_X.
\end{equation}
When $u(t)=v(t)$, we understand the right hand side through the limit in $(\ref{RHS})$. 
  \end{lemma}

  \begin{proof} 
By the definition of right semi-inner product, the right hand side of $(\ref{Dini-s-i})$ is:
\begin{equation}\label{RHS}
\lim_{h\to 0^+}\frac{1}{h}\left(\|(u-v)(t)+h(G(u(t))-G(v(t)))\|_X-\|(u-v)(t)\|_X\right),
\end{equation}
hence we just need to show that 
$$D^{+}\|(u-v)(t)\|_X=\lim_{h\to 0^+}\frac{1}{h}\left(\|(u-v)(t)+h(G(u(t))-G(v(t)))\|_X-\|(u-v)(t)\|_X\right).$$
Now using the definition of Dini derivative, we have:
\[
\begin{array}{rcl}
D^{+}\|(u-v)(t)\|_X&=&\displaystyle\limsup_{h\to 0^+}\frac{1}{h}\;\left(\|(u-v)(t+h)\|_X-\|(u-v)(t)\|_X\right)\\
&=&\displaystyle\limsup_{h\to 0^+}\frac{1}{h}\left(\|(u-v)(t)+h(\dot{u}-\dot{v})(t)+o(h)\|_X-\|(u-v)(t)\|_X\right)\\
&=&\displaystyle\limsup_{h\to 0^+}\frac{1}{h}\left(\|(u-v)(t)+h(\dot{u}-\dot{v})(t)\|_X-\|(u-v)(t)\|_X\right)\\
&=&\displaystyle\lim_{h\to 0^+}\frac{1}{h}\left(\|(u-v)(t)+h(\dot{u}-\dot{v})(t)\|_X-\|(u-v)(t)\|_X\right)\\
&=&\displaystyle\lim_{h\to 0^+}\frac{1}{h}\left(\|(u-v)(t)+h(G(u(t))-G(v(t)))\|_X-\|(u-v)(t)\|_X\right),
\end{array}
\]
where $\dot{u}=\displaystyle\frac{du}{dt}$.
Note that the fourth equality holds because of Remark $\ref{existence-of-norm}$. 
 \end{proof} 
 
  \begin{corollary} \label{key-1}
   Under the assumptions of  Lemma $\ref{Dini}$, for any $t\in [0, \infty)$ we have:
 \begin{equation}\label{Dini-u-b}
D^+\|u(t)-v(t)\|_X\leq M^+_{Y, X}[G]\|u(t)-v(t)\|_X.
\end{equation}
  \end{corollary}
   \begin{proof}
By the definition of the strong least upper bound logarithmic Lipschitz constant, $$\frac{((u-v)(t), G(u(t))-G(v(t)))_+}{\|(u-v)(t)\|_X^2}\leq M^+_{Y, X}[G].$$ Now apply Lemma $\ref{Dini}$ to the above inequality.
 \end{proof}

 \begin{corollary} \label{key1}
 Under the assumptions of Lemma $\ref{Dini}$, for any $t\in [0, \infty)$ we have:
  $$\|u(t)-v(t)\|_X\leq e^{M^+_{Y, X}[G]t}\|u(0)-v(0)\|_X.$$
 \end{corollary}
 
 \begin{proof}
 Apply Gronwall's inequality, \cite{Evans}, to $(\ref{Dini-u-b})$.
 \end{proof}
 
  \begin{remark}\label{M+Dini-linear}
Note that for any bounded linear operator $\mathcal{L}\colon X\to X$, and any solution $u\colon [0,T)\to X$ of $\displaystyle\frac{du}{dt}=\mathcal{L}u$, the above result says that 
 \begin{equation*}
 D^+\|u(t)\|_X=\displaystyle \frac{(u(t), \mathcal{L}u(t))_+}{\|u(t)\|_X^2}\|u(t)\|_X\leq M^+_X[\mathcal{L}]\|u(t)\|_X,
   \end{equation*}
  for all $t\in [0, T)$.
   \end{remark} 
  
\begin{lemma}\label{key4}
 Let $A_{p, Q}$, as defined above, denote an $n\times n$ diagonal matrix of operators on $\mathbf{Y}$ with the operators $d_i\Delta $ on the diagonal. Then $M^+_{\mathbf{Y}, \mathbf{X}}[A_{p, Q}]\leq 0$. 
\end{lemma}

\begin{proof}
To prove the lemma, we consider the following three cases:

{\bf{Case 1.}} $1<p<\infty$.
By the definition of $M^+_{\mathbf{Y}, \mathbf{X}}[A_{p, Q}]$, it's enough to show that for any $u\in \mathbf{Y}$ with $\normpQ{u}\neq0$, and any $\epsilon>0$, there exists $h_{\epsilon}>0$, depending on $\epsilon$, such that for $0<h<h_{\epsilon}$, 
\[\displaystyle\frac{1}{h}\left(\displaystyle\frac{\normpQ{u+hD\Delta u}}{\normpQ{u}}-1\right)=
\displaystyle\frac{1}{h}\left(\displaystyle\frac{\displaystyle\left(\sum_iq_i^p\|u_i+hd_i\Delta u_i\|_p^p\right)^{\frac{1}{p}}}{\left(\displaystyle\sum_iq_i^p\|u_i\|_p^p\right)^{\frac{1}{p}}}-1\right)<\epsilon.\]
(As $A_{p, Q}u=D\Delta u$, we write $D\Delta u$ instead of $A_{p, Q}u.$)

Therefore we'll show that for $h$ small enough
\begin{equation}\label{k(h)}
\displaystyle\sum_iq_i^p\|u_i+hd_i\Delta u_i\|_p^p < (1+\epsilon h)^p\displaystyle\sum_iq_i^p\|u_i\|_p^p.
\end{equation}
Let's define $k\colon[0,1]\to\r$ as follows: $$k(h)=\displaystyle\sum_iq_i^p\|u_i+hd_i\Delta u_i\|_p^p - (1+\epsilon h)^p\displaystyle\sum_iq_i^p\|u_i\|_p^p.$$
  Observe that $k$ is continuously differentiable:
\[
\begin{array}{lcl}
k'(h)&=&\displaystyle\frac{d}{dh}\displaystyle\sum_iq_i^p\int_{\Omega}\abs{u_i(\omega)+hd_i\Delta u_i(\omega)}^p\;d\omega-p\epsilon(1+\epsilon h)^{p-1}\displaystyle\sum_iq_i^p\|u_i\|_p^p\\
&=&\displaystyle\sum_iq_i^p\displaystyle\int_{\Omega}p\abs{u_i(\omega)+hd_i\Delta u_i(\omega)}^{p-2}\left(u_i(\omega)+hd_i\Delta u_i(\omega)\right) d_i\Delta u_i(\omega)\;d\omega\\
&-&p\epsilon(1+\epsilon h)^{p-1}\displaystyle\sum_iq_i^p\|u_i\|_p^p.
\end{array}
\]
Note that in general $\abs{g}^p$ is differentiable for $p>1$ and its derivative is $p\abs{g}^{p-2}gg'$.
Now by Green's identity, the Neumann boundary condition, and by the assumption that $\displaystyle\sum_iq_i^p\|u_i\|_p^p\neq0$, it follows integrating by parts that:
\[
\begin{array}{lcl}
k'(0)&=&p\displaystyle\sum_iq_i^p\displaystyle\int_{\Omega}\abs{u_i(\omega)}^{p-2}u_i(\omega)d_i\Delta u_i(\omega)\;d\omega-p\epsilon\displaystyle\sum_iq_i^p\|u_i\|_p^p\\
&=&-p(p-1)\displaystyle\sum_iq_i^pd_i\displaystyle\int_{\Omega}\abs{u_i(\omega)}^{p-2}\nabla{u_i(\omega)}^2\;d\omega-p\epsilon\displaystyle\sum_iq_i^p\|u_i\|_p^p\\
&<&0 .
\end{array}
\]
Since $k'(0)<0$ and $k'$ is continuous and $k(0)=0$, $k(h)< 0$ for $h$ small enough and therefore Inequality $(\ref{k(h)})$ holds.
 
Note that by the definition of $\mathbf{Y}$, any $u\in \mathbf{Y}$ satisfies the Neumann boundary condition.

{\bf{Case 2.}} $p=1$. Let 
\[g(p):=\displaystyle\lim_{h\to 0^+}\displaystyle\frac{1}{h}\left(\displaystyle\frac{\displaystyle\left(\sum_iq_i^p\|u_i+hd_i\Delta u_i\|_p^p\right)^{\frac{1}{p}}}{\left(\displaystyle\sum_iq_i^p\|u_i\|_p^p\right)^{\frac{1}{p}}}-1\right).\] Since $g(p)$ is a continuous function at $p=1$, and since in Case $1$, we showed that $g(p)\leq0$ for any $p>1$, we conclude that $g(1)\leq0$.

{\bf{Case 3.}} $p=\infty$. Before proving this case we need the following lemma, which is an easy exercise in
real analysis. (For completeness, we include a proof in an appendix.)

\begin{lemma}\label{p-limit}
Let $\Omega\subset\r^m$ be a Lebesgue measurable set with finite measure $\abs{\Omega}$ and let $f$ be a bounded, continuous function on $\r$. Then 
$F(p):=\displaystyle\left(\frac{1}{\abs{\Omega}}\int_{\Omega}\abs{f}^p\right)^{\frac{1}{p}}$ is an increasing function of $p$ and its limit as $p\to\infty$ is $\|f\|_{\infty}$.
\end{lemma}

For a fixed $p_0>1$, pick $u\in \mathbf{Y}$ with $\|u\|_{p_0,Q}\neq0$. By the definition of the norm, $\|u\|_{p_0,Q}\neq0$ implies that for some $i_0\in\{1,\cdots,n\}$, $\|u_{i_0}\|_{p_0}\neq0$. Let $\varphi(p):=\displaystyle\frac{1}{|\Omega|^{\frac{1}{p}}}\|u_{i_0}\|_p.$ By Lemma \ref{p-limit}, $\varphi$ is an increasing function of $p$, hence for any $p>p_0$, $\|u_{i_0}\|_{p}\geq \|u_{i_0}\|_{p_0}>0$. Now fix $i\in\{1,\cdots,n\}$, $p>p_0$, and $\epsilon>0$. Define $k$ as follows:
\[
k(h)=\displaystyle\left\{\begin{array}{ccc}\|u_i+hd_i\Delta u_i\|_p^p-(1+\epsilon h)^p\|u_{i_0}\|_p^p & & \mbox{if}\;\; \|u_{i_0}\|_{p}\geq\|u_{i}\|_{p}\\\|u_i+hd_i\Delta u_i\|_p^p-(1+\epsilon h)^p\|u_i\|_p^p & & \mbox{if} \;\;\|u_{i_0}\|_{p}\leq\|u_{i}\|_{p}\end{array}\right.
\]
 In both cases $k(0)\leq0$ and $k'(0)<0$ (the proof is similar to the proof of $k'(0)<0$ in Case $1$, since both $\|u_{i_0}\|_p>0$ and $\|u_{i_0}\|_p>0$).
 Therefore, for some small $h$, $k(h)\leq0$ which implies that:
 \[
 \displaystyle\lim_{h\to0^+}\frac{1}{h}\left(\displaystyle\frac{\|u_i+hd_i\Delta u_i\|_p}{\|u_i\|_p}-1\right)\leq0.
 \]
 Now by Lemma \ref{p-limit}, since $\displaystyle\frac{1}{|\Omega|^{\frac{1}{p}}}\|u_i+hd_i\Delta u_i\|_{p}\to \|u_i+hd_i\Delta u_i\|_{\infty}$, and $\displaystyle\frac{1}{|\Omega|^{\frac{1}{p}}}\|u_i\|_{p}\to\|u_i\|_{\infty}$ as $p\to\infty$, we can conclude that 
 \[
 \displaystyle\lim_{h\to0^+}\frac{1}{h}\left(\displaystyle\frac{\|u_i+hd_i\Delta u_i\|_{\infty}}{\|u_i\|_{\infty}}-1\right)\leq0.
 \]
In other words, for a fixed $\epsilon>0$, there exists $h_i>0$ such that for any $0<h<h_i$,
 \[\|u_i+hd_i\Delta u_i\|_{\infty}\leq(1+\epsilon h)\|u_i\|_{\infty}\quad\mbox{for any}\quad i\in\{1,\ldots,n\}.\]
  Let $h_0=\displaystyle\min_i{h_i}$. Then for any $0<h<h_0$,
 \[\displaystyle\max_{i}q_i\|u_i+hd_i\Delta u_i\|_{\infty}\;=:\;q_j\|u_j+hd_j\Delta u_j\|_{\infty}\;\leq\; q_j(1+\epsilon h)\|u_j\|_{\infty}\;\leq\;(1+\epsilon h)\displaystyle\max_iq_i\|u_i\|_{\infty},\]
which implies 
\[\displaystyle\lim_{h\to0^+}\frac{1}{h}\left(\frac{\displaystyle\max_{i}q_i\|u_i+hd_i\Delta u_i\|_{\infty}}{\displaystyle\max_iq_i\|u_i\|_{\infty}}-1\right)\leq0.\]
\end{proof}

\begin{lemma}\label{MM}
For any function $F$, any $1\leq p\leq\infty$, and any positive diagonal matrix $Q$, 
\[M^+_{\mathbf{Y}, \mathbf{X}}[\tilde{F}]\leq \MpQ[F],\]
where $\MpQ$ is the lub logarithmic Lipschitz constant induced by the norm $\|\cdot\|_{p,Q}$ defined on $\r^n$: $\|x\|_{p,Q}=\|Qx\|_p$.
\end{lemma}

\begin{proof}
 By the definition of $c:=\MpQ[F]$, we have 
\[
\lim_{h\to0^+}\frac{1}{h}\sup_{x\neq y\in V}\left(\frac{\normpQ{x-y+h(F(x)-F(y))}}{\normpQ{x-y}}-1\right)= c.
\]
Fix an arbitrary $\epsilon>0$. Then there exists $h_0>0$ such that for all $0<h<h_0$, 
\[\frac{1}{h}\sup_{x\neq y\in V}\left(\frac{\normpQ{x-y+h(F(x)-F(y))}}{\normpQ{x-y}}-1\right)<c+\epsilon.
\]
Therefore, for any $x\neq y$, and $0<h<h_0$
\begin{equation}\label{h0}
\frac{\normpQ{x-y+h(F(x)-F(y))}}{\normpQ{x-y}}<(c+\epsilon)h+1.
\end{equation}
For fixed $u\neq v\in \mathbf{Y}$, let $\Omega_1=\{\omega\in\bar\Omega:u(\omega)\neq v(\omega)\}$. 
Fix $\omega\in\Omega_1$, and let $x=u(\omega)$ and $y=v(\omega)$. We give a proof for the case $p<\infty$; the case $p=\infty$ is analogous. Using equation $(\ref{h0})$, we have:
\begin{equation}\label{normdef}
\displaystyle\frac{\left(\displaystyle\sum_iq_i^p\abs{u_i(\omega)-v_i(\omega)+h(F_i(u(\omega))-F_i(v(\omega)))}^p\right)^{\frac{1}{p}}}{\left(\displaystyle\sum_iq_i^p\abs{u_i(\omega)-v_i(\omega)}^p\right)^{\frac{1}{p}}}<(c+\epsilon)h+1.
\end{equation}
Multiplying both sides by the denominator and raising to the power $p$, we have:
\begin{equation}\label{forw}
\displaystyle\sum_iq_i^p\left|{u_i(\omega)-v_i(\omega)+h\left(F_i(u(\omega))-F_i(v(\omega))\right)}\right|^p<((c+\epsilon)h+1)^p\displaystyle\sum_iq_i^p\abs{u_i(\omega)-v_i(\omega)}^p.
\end{equation}
Since $\tilde{F}(u)(\omega)=F(u(\omega))$, Equation (\ref{forw}) can be written as:
\begin{equation}\label{ftilde}
\displaystyle\sum_iq_i^p\left|{u_i(\omega)-v_i(\omega)+h\left(\tilde{F}_i(u)(\omega)-\tilde{F}_i(v)(\omega)\right)}\right|^p<((c+\epsilon)h+1)^p\displaystyle\sum_iq_i^p\abs{u_i(\omega)-v_i(\omega)}^p.
\end{equation}
Now by taking the integral over $\bar\Omega$, using Lemma $\ref{U=S}$, we get:
\[\normpQ{u-v+h\left(\tilde{F}(u)-\tilde{F}(v)\right)}<((c+\epsilon)h+1)\normpQ{u-v}.
\]  
(Note that for $\omega\notin\Omega_1$, 
\[((c+\epsilon)h+1)^p\displaystyle\sum_iq_i^p\abs{u_i(\omega, t)-v_i(\omega, t)}^p=0\] which we can add to the right hand side of $(\ref{ftilde})$, and also 
\[
\displaystyle\sum_iq_i^p\abs{u_i(\omega)-v_i(\omega)+h(F_i(u(\omega))-F_i(v(\omega)))}^p=0
\]
which we can add to the left hand side of $(\ref{ftilde})$, and hence we can indeed take the integral over all $\bar\Omega$.)

Hence,
\[\displaystyle\lim_{h\to0^+}\frac{1}{h}\left(\frac{\normpQ{u-v+h\left(\tilde{F}(u)-\tilde{F}(v)\right)}}{\normpQ{u-v}}-1\right)\leq c+\epsilon.
\]
Now by letting $\epsilon\to0$ and taking $\sup$ over $u\neq v\in \mathbf{Y}$, we get $M^+_{\mathbf{Y}, \mathbf{X}}[\tilde{F}]\leq c$.
\end{proof}

\begin{proofoftheorem0}
 For any $1\leq p\leq\infty$, by subadditivity of semi inner product, Lemma \ref{key4}, and Lemma \ref{MM},
  \[M^+_{\mathbf{Y},\mathbf{X}}[\tilde{F}+A_{p, Q}]\leq M^+_{\mathbf{Y}, \mathbf{X}}[\tilde{F}]\leq c.\]
Now using Corollary \ref{key1}, 
\[\normpQ{\hat{u}(t)-\hat{v}(t)}\;\leq\; e^{ct}\normpQ{\hat{u}(0)-\hat{v}(0)},\]
for all $t\in[0,\infty).$
\end{proofoftheorem0}\qed

\begin{theorem}\label{thm-for-mu}
Consider the reaction-diffusion system $(\ref{re-di})-(\ref{i-c})$ and suppose Assumption \ref{as-pde} holds. In addition suppose for some $1\leq p\leq\infty$, $c\in\r$, and a positive diagonal matrix $Q$, $\mu_{p,Q}(J_F(x))\leq c$ for all $x\in V$, where $\mu_{p,Q}$ is the logarithmic norm induced by $\normpQ{\cdot}$. Then, for any two solutions $u, v$ of $(\ref{re-di}) - (\ref{i-c})$, we have 
\[\normpQ{u(\cdot, t)-v( \cdot, t)}\;\leq\; e^{ct}\normpQ{u(\cdot,0)-v(\cdot,0)}.\]
\end{theorem}

\begin{remark}\label{tv}
In general the result of Theorem \ref{thm-for-mu} holds also for time varying systems:
\begin{equation} \label{t-v}
\displaystyle\frac{\partial u}{\partial t}(\omega, t)=F(u(\omega, t), t)+D\Delta u(\omega, t)
\end{equation}
when we assume $\mu_{p,Q}(J_F(x, t))\leq c$, where $J_F(x, t)$ is the Jacobian $\displaystyle\frac{\partial F}{\partial x}(x, t)$, for all $x\in V$ and $t\geq0.$

We omit the details of this easy generalization. 
\end{remark}

\newtheorem*{proofoftheorem2}{Proof of Theorem \ref{thm-for-mu}}
To prove Theorem \ref{thm-for-mu}, we use the following proposition, from \cite{Soderlind2}.

\begin{proposition}\label{key3}
Let $(X, \|\cdot\|_X)$ be a normed space and $Y$ is a connected subset of $X$. Then for any (globally) Lipschitz and continuously differentiable function $f\colon Y\to \r^n$,
 $$\displaystyle\sup_{x\in Y}\mu_X(J_f(x))\leq M_{Y,X}[f].$$
 Moreover if $Y$ is convex, then  $$\displaystyle\sup_{x\in Y}\mu_X(J_f(x))= M_{Y,X}[f].$$
\end{proposition}

\begin{proofoftheorem2}
The proof is immediate from Theorem $\ref{main-result0}$ and Proposition $\ref{key3}$.
\end{proofoftheorem2}

\begin{corollary}\label{contracting}
Consider the reaction-diffusion system $(\ref{re-di})-(\ref{i-c})$ and suppose Assumption \ref{as-pde} holds. In addition suppose for some $1\leq p\leq\infty$, and a positive diagonal matrix $Q$, $\mu_{p,Q}(J_F(x))<0$ for all $x\in V$. Then $(\ref{re-di}) - (\ref{i-c})$ is contracting in $\mathbf{Y}$, meaning that solutions converge (exponentially) to each other, as $t\to+\infty$.
\end{corollary}

\section{Example}
We provide an example of a biochemical model which can be shown to be contractive by applying Corollary \ref{contracting} when using a weighted $L^1$ norm, but which is not contractive using any weighted $L^2$ norm, so that previous results can not be applied. Even more interestingly, this system is not contractive in any $L^p$ norm, $p>1$. The example is of great interest in molecular systems biology \cite{ddv07}, and contractivity in a weighted $L^1$ norm was shown for ODE systems in \cite{Russo}, but the PDE case was open. The variant with more enzymes discussed in \cite{Russo} can also be extended to the PDE case in an analogous fashion. 

 \textbf{Example $1$}. A typical biochemical reaction is one in which an enzyme $X$ (whose concentration is quantified by the non-negative variable $x=x(t)$) binds to a substrate $S$ (whose concentration is quantified by $s=s(t)\geq 0$), to produce a complex $Y$ (whose concentration is quantified by $y=y(t)\geq0$), and the enzyme is subject to degradation and dilution (at rate $\delta x$, where $\delta>0$) and production according to an external signal $z$, which we assume constant (a similar result would apply if $z$ is time dependent, see Remark \ref{tv}). An entirely analogous system can be used to model a transcription factor binding to a promoter, as well as many other biological process of interest. The complete system of chemical reactions is given by:
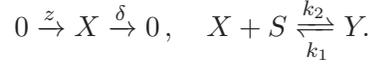
\[
0 \arrowchem{z} X \arrowchem{\delta} 0\,,\quad
X+S \arrowschem{k_1}{k_2} Y.
\]
We let the domain $\Omega$ represent the part of the cytoplasm where these chemicals are free to diffuse. Taking equal diffusion constants for $S$ and $Y$ (which is reasonable since typically $S$ and $Y$ have approximately the same size), a natural model is given by a reaction diffusion system 
\[
\begin{array}{lcl}\label{example}
x_t=z-\delta x+k_1y-k_2 s x+d_1\Delta x\\
y_t=-k_1y+k_2sx+d_2\Delta y\\
s_t=k_1y-k_2sx+d_2\Delta s.
\end{array}
\]
If we assume that initially $S$ and $Y$ are uniformly distributed, it follows that $\displaystyle\frac{\partial}{\partial t}\left(y(\omega, t)+s(\omega, t)\right)=0$, so $y(\omega, t)+s(\omega, t)=y(\omega, 0)+s(\omega, 0)=S_Y$ is a constant. Thus we can study the following reduced system: 
\[
\begin{array}{lcl}\label{example}
x_t=z-\delta x+k_1y-k_2(S_Y-y)x+d_1\Delta x\\
y_t=-k_1y+k_2(S_Y-y)x+d_2\Delta y.
\end{array}
\]
Note that $(x(t), y(t))\in V=[0, \infty)\times[0, S_Y]$ for all $t\geq0$ ($V$ is convex), and $S_Y$, $k_1$, $k_2$, $\delta$, $d_1$, and $d_2$ are arbitrary positive constants.

Let $J$ be the Jacobian of $F=(z-\delta x+k_1y-k_2(S_Y-y)x, -k_1y+k_2(S_Y-y)x)^T$:
$$J=\left(\begin{array}{cc}-\delta-k_2(S_Y-y) & k_1+k_2x \\k_2(S_Y-y) & -(k_1+k_2x)\end{array}\right).$$
In \cite{Russo}, 
it has been shown that $\displaystyle\sup_{x,y\in V}\mu_{1, Q}(J(x,y))<0$ for $Q=\diag\left(1, 1+\displaystyle\frac{\delta}{k_2S_Y}-\zeta\right)$, where $0<\zeta<\displaystyle\frac{\delta}{k_2S_Y}$. Therefore by Corollary \ref{contracting}, the system is contracting. Note that a $\emph{weighted}$ norm $L^1$ is necessary, since with $Q=I$ we obtain $\mu_1=0.$ 
 
  We'll show that for any $p>1$ and any diagonal $Q$, it is not true that $\mu_{p,Q}(J(x,y))<0$ for all $(x,y)\in V$.
  
  Without loss of generality we assume $Q=\diag(1,q)$. Then $$QJQ^{-1}=\displaystyle\left(\begin{array}{cc}-\delta-a & \displaystyle\frac{b}{q} \\aq & -b\end{array}\right),$$
 where $a=k_2(S_Y-y)\in[0,k_2S_Y]$ and $b=k_1+k_2x\in[k_1,\infty)$.
 
 We first consider the case $p\neq\infty$. We'll show that there exists $(x_0,y_0)\in C$ such that for any small $h>0$, $\|I+hQJ(x_0,y_0)Q^{-1}\|_p>1$. This will imply $\mu_{p,Q}(J(x_0,y_0))\geq 0.$ Computing explicitly, we have:
   \[
  \begin{array}{lcl}
  \|I+hQJQ^{-1}\|_p&=&\displaystyle\sup_{(\xi_1,\xi_2)\neq(0,0)}\frac{\left(\left|{\xi_1-h(\delta+a)\xi_1+h\displaystyle\frac{b\xi_2}{q}}\right|^p+\abs{haq\xi_1+\xi_2-hb\xi_2}^p\right)^{\frac{1}{p}}}{(\abs{\xi_1}^p+\abs{\xi_2}^p)^{\frac{1}{p}}}\\
  &\geq&\displaystyle\frac{\left(\left|{1-h(\delta+a)+h\displaystyle\frac{b\lambda}{q}}\right|^p+\abs{haq+\lambda-hb\lambda}^p\right)^{\frac{1}{p}}}{(1+\abs{\lambda}^p)^{\frac{1}{p}}},
  \end{array}
  \]
  where we take a point of the form $(\xi_1,\xi_2)=(1,\lambda)$, for a $\lambda>0$ which will be determined later.
  To show 
  \[
  \displaystyle\frac{\left(\left|{1-h(\delta+a)+h\displaystyle\frac{b\lambda}{q}}\right|^p+\left|{haq+\lambda-hb\lambda}\right|^p\right)^{\frac{1}{p}}}{(1+\abs{\lambda}^p)^{\frac{1}{p}}}>1,
  \]
  we'll equivalently show that for any small enough $h>0$:
  \begin{equation}\label{example1}
\displaystyle \frac{1}{h}\left( \left|{1-h(\delta+a)+h\displaystyle\frac{b\lambda}{q}}\right|^p+\abs{haq+\lambda-hb\lambda}^p-1-\abs{\lambda}^p\right)>0.
  \end{equation}
  Note that the $\displaystyle\lim_{h\to0^+}$ of the left hand side of the above inequality is $f'(0)$ where 
  \[
 f(h)\;=\;\abs{1+h(\frac{b\lambda}{q}-(\delta+a))}^p+\abs{\lambda+h(aq-b\lambda)}^p.
  \]
   Therefore it suffices to show that $f'(0)>0$ for some value $(x_0,y_0)\in C$ (because $f'(0)>0$ implies that there exists $h_0>0$ such that for $0<h<h_0$, $(\ref{example1})$ holds). 
   Since $p>1$, by assumption, $f$ is differentiable and 
  \[
    \begin{array}{lcl}
  f'(h)&=&\displaystyle p\left(\frac{b\lambda}{q}-(\delta+a)\right)\left|{1+h\left(\frac{b\lambda}{q}-(\delta+a)\right)}\right|^{p-2}\left(1+h\left(\frac{b\lambda}{q}-(\delta+a)\right)\right)\\
  &+& p(aq-b\lambda)\left|{\lambda+h(aq-b\lambda)}\right|^{p-2}\left(\lambda+h(aq-b\lambda)\right).
    \end{array}
 \]
  Hence, since $\lambda>0$
  \[
  \begin{array}{lcl}
  f'(0)&=&p\left(\displaystyle\frac{b\lambda}{q}-(\delta+a)\right)+ p(aq-b\lambda)\lambda^{p-1}\\
  &=&p\left(\displaystyle\frac{b\lambda}{q}-a\right)(1-\lambda^{p-1}q)-p\delta.
  \end{array}
  \]
  Choosing $\lambda$ small enough such that $1-\lambda^{p-1}q>0$ and choosing $x$, or equivalently $b$, large enough, we can make $f'(0)>0$. 
  
  For $p=\infty$, using Table \ref{tab-mu}, $\mu_p(QJQ^{-1})=\displaystyle\max\left\{-\delta-a+\displaystyle\frac{b}{q},-b+aq\right\}$. For large enough $x$, $-\delta-a+\displaystyle\frac{b}{q}>0$ (and $-b+aq<0$) and hence $\mu_{\infty}(QJQ^{-1})>0$.
  
\section{Diffusive interconnection of ODEs}
In this section, we derive a result analogous to that for PDE's for a network of identical ODE models which are diffusively interconnected. We study systems of ODE's as follows:
\begin{equation}\label{discrete}
\dot{u}(t)=\tilde{F}(u(t))-(L\otimes D)u(t).
\end{equation}
\begin{assumption}\label{as-ode}
In $(\ref{discrete})$, we assume:
\begin{itemize}
\item For a fixed convex subset of $\r^n$, say $V$, $\tilde{F}\colon V^{N}\to\r^{nN}$ is a function of the form:
\[\tilde{F}(u)=\left(F(u_1)^T, \cdots, F(u_N)^T\right)^T,\]
 where $u=\left(u_1^T,\cdots, u_N^T\right)^T$, with $u_i\in V$ for each $i$, and $F\colon V\to\r^n$ is a (globally) Lipschitz function. 
\item For any $u\in V^N$ we define $\normpQ{u}$ as follows:
\[\normpQ{u}=\left\|\left(\|Qu_1\|_p, \cdots, \|Qu_N\|_p\right)^T\right\|_p,\]
where $Q=\diag{(q_1,\cdots, q_n)}$ is a positive diagonal matrix and $1\leq p\leq\infty$.

With a slight abuse of notation, we use the same symbol for a norm in $\r^n$: \[\|x\|_{p,Q}:=\|Qx\|_p.\]

\item $u\colon [0,\infty)\to V^N$ is a continuously differentiable function.
\item $D=\diag(d_1,\cdots, d_n)$ with $d_i>0$, which we call the diffusion matrix.
\item $L\in\r^{N\times N}$ is a symmetric matrix and $L\mathbf{1}=0$, where $\mathbf{1}=(1,\cdots, 1)^T$. 
We think of $L$ as the Laplacian of a graph that describes the interconnections among component subsystems. 
\end{itemize}
\end{assumption}

\begin{theorem}\label{ODE}
Consider the system $(\ref{discrete})$ and suppose Assumption \ref{as-ode} holds. Let $c=\MpQ [F]$, where $\MpQ$ is the lub logarithmic Lipschitz constant induced by the norm $\|\cdot\|_{p,Q}$ on $\r^n$  defined by $\|x\|_{p,Q}:=\|Qx\|_p$. Then for any two solutions $u, v$ of $(\ref{discrete})$, we have 
$$\normpQ{u(t)-v(t)}\leq e^{ct}\normpQ{u(0)-v(0)}.$$
\end{theorem}

This theorem is proved by following the same steps as in the PDE case and using Lipschitz norms and properties of discrete Laplacians on finite graphs. For ODEs, we can make some of the steps more explicit, and for purposes of exposition, we do so next. We start with several technical lemmas.
  
The following elementary property of logarithmic norms is well-known. To see more properties of logarithmic norm see \cite{Desoer}.

\begin{lemma}\label{mu-prop}
Let $\lambda$ be the largest real part of an eigenvalue of $A$. Then, $\mu_{p,Q}(A)\geq\lambda$. 
\end{lemma}

\begin{proposition}\label{negM+p}
For any $1\leq p\leq\infty$, $M^+_p(-L\otimes D)=0$, where $M^+_p$ is the strong least upper bound logarithmic Lipschitz constant induced by the $L^p$ norm.
\end{proposition}

\begin{proof}
Let $\mathcal{L}=-L\otimes D=(\mathcal{L}_{ij})$. Note that since $L\mathbf{1}=0$, by the definition of  Kronecker product, $\mathcal{L}\mathbf{1}=0$. In addition because $L$ is symmetric and $D$ is diagonal, $\mathcal{L}$ is also symmetric and therefore $\mathcal{L}\mathbf{1}=\mathbf{1}\mathcal{L}=0$. Also the off diagonal entries of $\mathcal{L}$, like $-L$, are positive because $L$ is a Laplacian matrix. By Corollary \ref{mu=M+}, it suffices to show that $\mu_p(\mathcal{L})=0$ for any $p$.
We first show that $\mu_p(\mathcal{L})=0$ for $p=1,\infty$. For $p=1$,
 \[\mu_1(\mathcal{L})=\displaystyle\max_j\sum_{i\neq j,i=1,\dots,nN}(\mathcal{L}_{ii}+|\mathcal{L}_{ij}|)=\max_j{0}=0.\]
 Similarly for $p=\infty$, 
 \[\mu_{\infty}(\mathcal{L})=\displaystyle\max_i\sum_{i\neq j,j=1,\dots,nN}(\mathcal{L}_{ii}+|\mathcal{L}_{ij}|)=\max_j{0}=0.\]
Now suppose $p\neq 1,\infty$.
By Lemma \ref{mu-prop}, $\mu_p(\mathcal{L})\geq \Re\lambda$, where $\lambda$ is an eigenvalue of $\mathcal{L}$. Because $\mathcal{L}\mathbf{1}=0$, $\lambda=0$ is an eigenvalue of $\mathcal{L}$; therefore $\mu_p(\mathcal{L})\geq 0$.
To show that $\mu_p(\mathcal{L})\leq 0$, by Remark \ref{M+Dini-linear}, it suffices to show that $D^{+}\|u\|_p\leq0$ where $u$ is the solution of $\dot{u}=\mathcal{L}u$. By the definition of Dini derivative, it suffices to show that $\|u(t)\|_p$ is a non-increasing function of $t$. Let $\Phi(u(t)):=\|u(t)\|_p^p$, where $u=(u_1^T,\cdots, u_N^T)^T$ with $u_i=(u_i^1,\cdots,u_i^n)^T\in V^n$. Here we abuse the notation and assume that $u=(u_1,\cdots, u_{nN})^T$. We'll show that $\displaystyle\frac{d\Phi}{dt}(u(t))\leq 0$.

First we'll prove the following inequality:
\begin{lemma}\label{ab}
For any real $\alpha$ and $\beta$ and $1\leq p$:
$$(|\alpha|^{p-2}+|\beta|^{p-2})\alpha\beta\leq |\alpha|^p+|\beta|^p.$$
\end{lemma}

\begin{proof}
For $\alpha\beta\leq0$, the inequality is trivial. Suppose $\alpha\beta>0$, and w.l.o.g $|\beta|\geq|\alpha|$ and let $\lambda=\displaystyle\frac{\beta}{\alpha}$. Then it suffices to prove that for $\lambda\geq1$, 
$$(1+\lambda^{p-2})\lambda\leq\lambda^p+1.$$
Let $f(\lambda)=\lambda^{p-1}+\lambda-\lambda^{p}-1$. We want to show that $f(\lambda)\leq0$ for $\lambda\geq1$. Since $f(1)=f'(1)=0$ and $f''(\lambda)\leq0$ for $\lambda\geq1$, indeed $f(\lambda)\leq0$.
\end{proof}

As we explained above, $\mathcal{L}$ is symmetric and $\mathcal{L}\mathbf{1}=0$. Using this information and the above inequality:
\[
\begin{array}{rcl}
\displaystyle\frac{d\Phi}{dt}(u(t))&=&\displaystyle\sum_{i=1}^{nN} \frac{d\Phi}{d u_i}\frac{d u_i}{d t}\\
&=&\displaystyle\triangledown \Phi \cdot \dot{u}\\
&=&\displaystyle\triangledown \Phi \cdot \mathcal{L}u\\
&=&\displaystyle p(|u_1|^{p-2}u_1,\cdots, |u_{nN}|^{p-2}u_{nN})\mathcal{L} (u_1, \cdots, u_{nN})^T\\
&=&\displaystyle p\sum_{i,j}|u_i|^{p-2}u_i\mathcal{L}_{ij}u_j\\
&=&\displaystyle p\sum_{i}|u_i|^p\mathcal{L}_{ii}+p\sum_{i<j}\mathcal{L}_{ij}(|u_i|^{p-2}+|u_j|^{p-2})u_iu_j\\
&\leq& \displaystyle p\sum_{i}|u_i|^p\mathcal{L}_{ii}+p\sum_{i<j}\mathcal{L}_{ij}(|u_i|^{p}+|u_j|^{p})\\
&=&\displaystyle p\sum_{i}|u_i|^p\mathcal{L}_{ii}+p\sum_{i\neq j}\left(\mathcal{L}_{ij}|u_i|^{p}+\mathcal{L}_{ji}|u_j|^{p}\right)\\
&=&\displaystyle p\sum_{i}|u_i|^p\left(\mathcal{L}_{ii}+\sum_{i\neq j}\mathcal{L}_{ij}\right)\\
&=&0,
\end{array}
\]
since $\displaystyle\frac{\partial \Phi}{\partial u_i}=\displaystyle\frac{\partial}{\partial u_i}|u_i|^{p}=p|u_i|^{p-1}\displaystyle\frac{u_i}{|u_i|}=p|u_i|^{p-2}u_i$. Note that $|x|^p$ is differentiable for $p>1$.
\end{proof}

\begin{lemma}\label{muX=mup} Let $\mu_p$ and $\mu_{p,Q}$ denote the logarithmic norms induced by $\|\cdot\|_p$ and $\|\cdot\|_{p,Q}$ respectively. Then 
\[
\mu_{p,Q}(-L\otimes D)=\mu_p(-L\otimes D).
\]
 \end{lemma}

\begin{proof}
Recall the following properties of Kronecker product:
\begin{itemize}
\item $(A\otimes B)(C\otimes D)=(AC)\otimes (BD)$;
\item If $A$ and $B$ are invertible, then $(A\otimes B)^{-1}=A^{-1}\otimes B^{-1}$.
\end{itemize}
Hence:
\[
\begin{array}{rcl}
\displaystyle\mu_{p, Q}(-L\otimes D)
&=&\displaystyle\mu_p[(I\otimes Q)(-L\otimes D)(I\otimes Q^{-1})]\\
&=&\displaystyle\mu_p(-L\otimes QDQ^{-1})\\
&=&\displaystyle\mu_p(-L\otimes D).
\end{array}
\]
The last equality holds because both $Q$ and $D$ are diagonal, and so are commutative. Therefore $QDQ^{-1}=DQQ^{-1}=D.$
\end{proof}

\begin{proposition}\label{negM+X} Let $M_{p,Q}^{+}$ denote the strong least upper bound logarithmic Lipschitz constant induced by the norm $\|\cdot\|_{p,Q}$ on $\r^{nN}$. Then, 
\[M_{p,Q}^{+}[-L\otimes D]=0.\]
 \end{proposition}

\begin{proof}
By Proposition \ref{negM+p}, Corollary \ref{mu=M+} and Lemma \ref{muX=mup},
\[
M_{p,Q}^{+}[-L\otimes D]\;=\;\displaystyle\mu_{p,Q}(-L\otimes D)
\;=\;\displaystyle\mu_p(-L\otimes D)
\;=\;\displaystyle M_{p}^{+}[-L\otimes D]
\;=\;0.
\]
\end{proof}

\begin{lemma}\label {MM2}
 Let $M^+_{p,Q}$ denote the strong lub logarithmic Lipschitz constant induced by the norm $\|\cdot\|_{p,Q}$ on $\r^{nN}$ and $\MpQ$ is the lub logarithmic Lipschitz constant induced by the norm $\|\cdot\|_{p,Q}$ on $\r^{n}$. Then, 
\[M^+_{p,Q}[\tilde{F}]\leq \MpQ[F].\]
\end{lemma}

\begin{proof}
The proof is exactly the same as the proof of Proposition \ref{MM}. 
\end{proof}

 $\textbf{Proof of theorem \ref{ODE}.}$ By subadditivity of $M^+_{p,Q}$, Proposition \ref{subadd}, Proposition \ref{negM+X}, and Lemma \ref{MM2}:
\[
\begin{array}{rcl}
M^+_{p,Q}[\tilde{F}-L\otimes D]\;\leq\;\displaystyle M^+_{p,Q}[\tilde{F}]+M^+_{p,Q}[-L\otimes D]
\;\leq\;M_{p,Q}[F]
\;=\;c.
\end{array}
\]
Now using Corollary \ref{key1}, 
$$\|u(t)-v(t)\|_{p,Q}\leq e^{ct}\|u(0)-v(0)\|_{p,Q}.$$\qed

\begin{lemma}\label{forCE}
Assume $F$ is a linear operator. Then 
\begin{equation}\label{forCE}
\mu_{p, Q}(\tilde{F}-L\otimes D)\leq \mu_{q, Q}(F) \quad\mbox{if} \quad p=q.
\end{equation}
\end{lemma}

 \begin{proof} 
 The proof is immediate by subadditivity of logarithmic norm, Proposition \ref{MM2}, and Corollary \ref{mu=M+}.
 \end{proof}

\begin{remark} 
Note that (\ref{forCE}) doesn't need to hold if $p\neq q$. Consider the following system:
 \[
\begin{array}{rcl}
\dot{x}_1&=&Ax_1+D(x_2-x_1)\\
\dot{x}_2&=&Ax_2+D(x_1-x_2),
\end{array}
\]
where $x_i\in\r^2$, $A=\mbox{\scriptsize%
$\displaystyle\begin{bmatrix}-2 & 1 \\1 & -2\end{bmatrix}$}$ and $D=\diag(d_1, d_2)$. In this example $L=\mbox{\scriptsize%
$\displaystyle\begin{bmatrix}1 & -1 \\-1 & 1\end{bmatrix}$}$ and $F=\diag(A, A)$. We'll show that for $Q=\diag{(3,1)}$, $\mu_{2, Q}(A)<0$ while $\mu_{1, Q}(F-L\otimes D)>0$.
 
 By Table $\ref{tab-mu}$, 
  \[
\begin{array}{lcl}
\mu_{2, Q}(A)=\mu_2(QAQ^{-1})=\mu_2\mbox{\scriptsize%
$\displaystyle\begin{bmatrix}-2 & 3 \\\frac{1}{3} & -2\end{bmatrix}$}<0.
\end{array}
\]
 \[
\begin{array}{lcl}
\mu_{1, Q}(F-L\otimes D)&=&\displaystyle\mu_{1, Q}
\mbox{\scriptsize%
$\displaystyle\begin{bmatrix}-2-d_1 & 1 & d_1 & 0 \\1 & -2-d_2 & 0 & d_2 \\d_1 & 0 & -2-d_1 & 1 \\0 & d_2 & 1 & -2-d_2\end{bmatrix}$}\\
&=&\displaystyle\mu_{1}\mbox{\scriptsize%
$\displaystyle\begin{bmatrix}-2-d_1 & 3 & d_1 & 0 \\\frac{1}{3} & -2-d_2 & 0 & d_2 \\d_1 & 0 & -2-d_1 & 3\\0 & d_2 & \frac{1}{3} & -2-d_2\end{bmatrix}$}\\
&=&1>0.
\end{array}
\]
\end{remark}

\begin{theorem}
Consider the reaction-diffusion ODE $(\ref{discrete})$ and suppose Assumption \ref{as-ode} holds. In addition assume that $F$ is continuously differentiable and $\mu_{p,Q}(J_F(x))\leq c$ for all $x\in V$. Then for any two solutions $u, v$ of $(\ref{discrete})$ we have 
$$\|u(t)-v(t)\|_{p,Q}\leq e^{ct}\|u(0)-v(0)\|_{p,Q}.$$
\end{theorem}

\begin{proof}
The proof is immediate by Theorem $\ref{ODE}$ and Proposition $\ref{key3}$.
\end{proof}

 \section{Synchronization}
In this section we will show that for any $1<p<\infty$, and $N=2$, and $3$ nodes which are interconnected according to a connected, undirected graph, if $c=\displaystyle\sup_{u\in V}\mu_{p,Q}(J_F(u)-\lambda D)$, where $V$ is a convex subset of $\r^n$ and $\lambda$ is the smallest positive eigenvalue of $-L$, then every solution $u=(u_1,\cdots,u_N)$ of (\ref{discrete}):
\begin{equation*}
\dot{u}(t)=\tilde{F}(u(t))-(L\otimes D)u(t)
\end{equation*}
has the following property:
$$W(t)\leq e^{ct}W(0),$$ 
where $W(t)=\|w(t)\|_p$ and $w$ is a row vector defined by 
\[
w:=\left(\scalefont{0.7}{\normpQ{u_1-u_2}, \ldots, \normpQ{u_1-u_N}},\; \scalefont{0.7}{\normpQ{u_2-u_3}, \ldots, \normpQ{u_2-u_N}},\;\cdots,\; \scalefont{0.7}{\normpQ{u_{N-1}-u_N}}\right).
\] 

To show this, we first state the following lemma from \cite{Deimling}:

 \begin{lemma}\label{Deimling}
 For any $p\in(1,\infty)$, $(u,v)_+=\|u\|_p^{2-p}\sum_{i=1}^n\abs{u_i}^{p-2}u_iv_i.$
 \end{lemma}
 
For $N=2$ nodes, there is just one possible graph: The complete graph with graph Laplacian matrix 
$L=\mbox{\scriptsize%
$\displaystyle\begin{bmatrix}-1 & 1 \\1 & -1\end{bmatrix}$}$ which results in the following ODE system:
  \[
\begin{array}{rcl}
\dot{x}=F(x)+D(y-x)\\
\dot{y}=F(y)+D(x-y).
\end{array}
\]
Note that the smallest and only positive eigenvalue of $-L$ is $\lambda=2$.
For a fixed solution $(x,y)$ of this ODE, let $w(t)=\normpQ{x-y}$, $W(t)=\|w(t)\|_p$, and $\tilde{D}(x):=Dx$. Then by the definition of $M^+_{p,Q}$, and Lemma \ref{Deimling} we have,
 \[
\begin{array}{lcl}
\displaystyle\frac{dW^p(t)}{dt}&=&p\sum_{i=1}^n\abs{q_i(x_i-y_i)}^{p-2}q_i(x_i-y_i)q_i(\dot{x}_i-\dot{y}_i)\\
&=&p\sum_{i=1}^n\abs{q_i(x_i-y_i)}^{p-2}q_i(x_i-y_i)q_i\left(F_i(x)-F_i(y)-2D_i(x_i-y_i)\right)\\
&\leq&p\|Q(x-y)\|_p^pM^+_{p,Q}[F-2\tilde{D}]\\
&=&pW^p(t)M^+_{p,Q}[F-2\tilde{D}]\\
&\leq&pW^p(t)\displaystyle\sup_{(x,y)\in V}\mu_{p,Q}(J_F(x,y)-2D).
\end{array}
\]
The last inequality results by Proposition \ref{key3}, $$\displaystyle\sup_{(x,y)\in V}\mu_{p,Q}(J_F(x,y)-2D)=M_{p,Q}[F-2\tilde{D}]\geq M^+_{p,Q}[F-2\tilde{D}].$$
Therefore $W(t)\leq e^{\mu_{p,Q}(J_F-2D)t}W(0)$, i.e., in this case, 
 $$\|x(t)-y(t)\|_{p,Q}\leq e^{ct}\|x(0)-y(0)\|_{p,Q},$$
 where $c=\displaystyle\sup_{(x,y)\in V}\mu_{p,Q}(J_F(x,y)-2D).$

For $N=3$, there are two possible graphs: 

First, the complete graph with graph Laplacian matrix $L=\mbox{\scriptsize%
$\displaystyle\begin{bmatrix}-2 & 1&1 \\1 & -2&1\\ 1 &1 &-2\end{bmatrix}$}$ which leads to the following ODE system:
  \[
\begin{array}{rcl}
\dot{x}=F(x)+D(y-2x+z)\\
\dot{y}=F(y)+D(x-2y+z)\\
\dot{z}=F(z)+D(x+y-2z).
\end{array}
\]
Note that the smallest positive eigenvalue of $-L$ is $\lambda=3$. For a fixed solution $(x,y,z)$ of this ODE, define $w$ and $W$ as follows:
 \[w(t)=\left(\|x-y\|_{p,Q}, \|y-z\|_{p,Q}, \|x-z\|_{p,Q}\right),\]
  and  
$W(t):=\|w(t)\|_p.$ Similar to case $N=2$, we have 
 \[
\begin{array}{lcl}
\displaystyle\frac{dW^p(t)}{dt}&\leq&pW^p(t)M^+_{p,Q}[F-3\tilde{D}]\\
&\leq&pW^p(t)\displaystyle\sup_{(x,y,z)\in V}\mu_{p,Q}\left(J_F(x,y,z)-3D\right), 
\end{array}
\]
which leads to 
$$W^p(t)\leq e^{pct}W^p(0),$$ where $c=\displaystyle\sup_{(x,y,z)\in V}\mu_{p,Q}\left(J_F(x,y,z)-3D\right).$
Taking the $p$th roots, 
\[ W(t)\leq e^{ct}W(0).\]

The second graph has the graph Laplacian matrix $L=\mbox{\scriptsize%
$\displaystyle\begin{bmatrix}-1 & 1&0 \\1 & -2&1\\ 0&1 &-1\end{bmatrix}$}$,
with the following ODE system:
  \[
\begin{array}{lcl}
\dot{x}=F(x)+D(y-x)\\
\dot{y}=F(y)+D(x-2y+z)\\
\dot{z}=F(z)+D(y-z).
\end{array}
\]
Note that the smallest positive eigenvalue of $-L$ is $\lambda=1$. Let again for a fixed solution $(x,y,z)$, $w(t)=(\|x-y\|_{p,Q},\|y-z\|_{p,Q},\|x-z\|_{p,Q})$ and $W(t)=\|w(t)\|_p$. Then 
\[
\begin{array}{lcl}
\displaystyle\frac{dW^p(t)}{dt}&=&p\sum_{i=1}^n\abs{q_i(x_i-y_i)}^{p-2}q_i(x_i-y_i)q_i(\dot{x}_i-\dot{y}_i)\\
&+&p\sum_{i=1}^n\abs{q_i(y_i-z_i)}^{p-2}q_i(y_i-z_i)q_i(\dot{y}_i-\dot{z}_i)\\
&+&p\sum_{i=1}^n\abs{q_i(x_i-z_i)}^{p-2}q_i(x_i-z_i)q_i(\dot{x}_i-\dot{z}_i)\\
&=&p\sum_{i=1}^n\abs{q_i(x_i-y_i)}^{p-2}q_i(x_i-y_i)q_i\left(F_i(x)-F_i(y)-2D_i(x_i-y_i)+D_i(y_i-z_i)\right)\\
&+&p\sum_{i=1}^n\abs{q_i(y_i-z_i)}^{p-2}q_i(y_i-z_i)q_i\left(F_i(y)-F_i(z)-2D_i(y_i-z_i)+D_i(x_i-y_i)\right)\\
&+&p\sum_{i=1}^n\abs{q_i(x_i-z_i)}^{p-2}q_i(x_i-z_i)q_i\left(F_i(x)-F_i(z)-D_i(x_i-z_i)\right).
\end{array}
\]
By Lemma \ref{ab}, 
\[
\sum_{i=1}^n\abs{x_i-y_i}^{p-2}(x_i-y_i)(y_i-z_i)+\sum_{i=1}^n\abs{y_i-z_i}^{p-2}(y_i-z_i)(x_i-y_i)\leq\sum_{i=1}^n\abs{x_i-y_i}^p +\sum_{i=1}^n\abs{y_i-z_i}^p.
\]
Hence 
\[
\begin{array}{lcl}
\displaystyle\frac{dW^p(t)}{dt}&\leq&p\sum_{i=1}^n\abs{q_i(x_i-y_i)}^{p-2}q_i(x_i-y_i)q_i\left(F_i(x)-F_i(y)-D_i(x_i-y_i)\right)\\
&+&p\sum_{i=1}^n\abs{q_i(y_i-z_i)}^{p-2}q_i(y_i-z_i)q_i\left(F_i(y)-F_i(z)-D_i(y_i-z_i)\right)\\
&+&p\sum_{i=1}^n\abs{q_i(x_i-z_i)}^{p-2}q_i(x_i-z_i)q_i\left(F_i(x)-F_i(z)-D_i(x_i-z_i)\right)\\
&\leq&pW^p(t)M^+_{p,Q}[F-\tilde{D}]\\
&\leq&pW^p(t)\displaystyle\sup_{(x,y,z)\in V}\mu_{p,Q}(J_F(x,y,z)-D).
\end{array}.
\]
Similar to the previous case: \[ W(t)\leq e^{ct}W(0),\]
where in this case $c=\displaystyle\sup_{(x,y,z)\in V}\mu_{p,Q}\left(J_F(x,y,z)-D\right).$

 \section{Appendix}
 \subsection{Proof of Theorem \ref{mu}.}
To prove the theorem we need the following
\begin{proposition}\label{Joo}\cite{Joo}
Let $X$, $Y$ be arbitrary sets and let $\varphi:X\times Y\to \r$ be an arbitrary function. For any $y\in Y$ and $c\in\r$, denote $H_{y,c}=\{x\in X: \varphi(x,y)\geq c\}$ and $\mathcal{C}$ the set of all real numbers $c$ such that for any $y\in Y$, $H_{y,c}\neq\emptyset$ and let $c^*=\sup\mathcal{C}$. Then $$(B=)\displaystyle\sup_{x\in X}\inf_{y\in Y}\varphi(x,y)=\displaystyle\inf_{y\in Y}\sup_{x\in X}\varphi(x,y)(=J)$$ if and only if for every $c<c^*$, $\displaystyle\bigcap_{y\in Y}H_{y,c}\neq\emptyset$. In this case $B=J=c^*$. 
\end{proposition}

A proof can be found in \cite{Joo}. We provide a proof next. First we need the following 
\begin{lemma}
For any fixed $y\in Y$, if $c_1<c_2$ then $H_{y,c_2}\subset H_{y,c_1}$.
\end{lemma}

\begin{proof}
Pick any $x\in H_{y,c_2}$. By the definition of $H_{y,c_1}$, $\varphi(x,y)\geq c_2>c_1$ and hence $x\in H_{y,c_1}$.
\end{proof}

\begin{corollary}\label{n.e.intersection}
If $c_1<c_2$ then $\displaystyle\bigcap_{y\in Y} H_{y,c_2}\subset\bigcap_{y\in Y} H_{y,c_1}$.
\end{corollary}

\newtheorem*{proofofProposition5}{Proof of Proposition \ref{Joo}}
\newtheorem*{proofoftheoremmu}{Proof of Theorem \ref{mu}}

\begin{proofofProposition5}
First we assume that for every $c<c^*$, $\displaystyle\bigcap_{y\in Y}H_{y,c}\neq\emptyset$. Using this assumption we'll show that $B=J=c^*$. To this end we'll show the following three inequalities:
\begin{enumerate}
\item $B\leq J$. 
\[
\begin{array}{rcl}
\varphi(x,y)&\leq&\displaystyle\sup_{x\in X}\varphi(x,y)\qquad\forall x,y\\
&\Rightarrow&\displaystyle\inf_{y\in Y}\varphi(x,y)\leq\inf_{y\in Y}\sup_{x\in X}\varphi(x,y)\qquad\forall x\\
&\Rightarrow&\displaystyle\sup_{x\in X}\inf_{y\in Y}\varphi(x,y)\leq\inf_{y\in Y}\sup_{x\in X}f(x,y).
\end{array}
\]
Hence by the definition, $B\leq J$.
\item $J\leq c^*$.
For an arbitrary $c>c^*$ we'll show that $J\leq c$ and then we can conclude that $J\leq c^*$. Since $c>c^*$, there exists $y_0\in Y$ such that $H_{y_0,c}=\emptyset$ (otherwise $c\in\mathcal{C}$ and so $c\leq c^*$). This means that for all $x\in X$, $\varphi(x,y_0)<c$ which implies $\displaystyle\sup_{x\in X}\varphi(x,y_0)\leq c$ and hence $J=\displaystyle\inf_{y\in Y}\sup_{x\in X}\varphi(x,y_0)\leq c$.

\item $c^*\leq B$. For an arbitrary $c<c^*$ we'll show that $B\geq c$ and then we can conclude that $B\geq c^*$. Since $c<c^*$, there exists $c_0\in\mathcal{C}$ such that $c<c_0$ (otherwise $c=\displaystyle\sup\mathcal{C}$). Since we assumed that for every $c<c^*$, $\displaystyle\bigcap_{y\in Y}H_{y,c}\neq\emptyset$, and since $c_0\leq c^*$, $c_0\in\mathcal{C}$, we have $\displaystyle\bigcap_{y\in Y} H_{y,c_0}\neq\emptyset$. By Corollary \ref{n.e.intersection}, because $c<c_0$, then $\displaystyle\bigcap_{y\in Y} H_{y,c}\neq\emptyset$, i.e. there exists $x_0\in X$ such that for all $y\in Y$, $\varphi(x_0,y)\geq c$, which implies $\displaystyle\inf_{y\in Y}\varphi(x_0,y)\geq c$ and hence $B=\displaystyle\sup_{x\in X}\inf_{y\in Y}\varphi(x,y)\geq c$.

\end{enumerate}

Now we suppose $B=J$ and fix $c<c^*$. We'll show that $\displaystyle\bigcap_{y\in Y}H_{y,c}\neq\emptyset$. Since $c<c^*$, there exists $c_0\in\mathcal{C}$ such that $c<c_0$. This means that there exists $x_0\in X$ such that for all $y\in Y$, $\varphi(x_0,y)\geq c_0$, i.e. $J\geq c_0$. Since $J=B$, there exists $x_1\in X$ such that for all $y\in Y$, $\varphi(x_1,y)\geq c_0>c$, i.e. $\displaystyle\bigcap_{y\in Y}H_{y,c}\neq\emptyset$.\qed
\end{proofofProposition5}

For a fixed arbitrary norm $\|\cdot\|$ on $\r^{nN}$ and a fixed arbitrary matrix $A\in\r^{nN\times nN}$, define $\varphi:S^{nN-1}\times (0,1)\to\r$ by $$\varphi(v,h)=\displaystyle\frac{1}{h}\left(\|v+hAv\|-1\right),$$ where $S^{nN-1}=\{v\in\r^{nN}: \|v\|=1\}$. For any $h\in(0,1)$ and $c\in\r$, let $H_{h,c}=\{v\in S^{nN-1}: \varphi(v,h)\geq c\}$ and let $\mathcal{C}$ be the set of all real numbers $c$ such that  $H_{h,c}\neq\emptyset$ whenever $h\in(0,1)$. Let $c^*=\sup\mathcal{C}$.

\begin{lemma}\label{f-dec}
$\varphi(v,h)=\displaystyle\frac{1}{h}(\|v+hAv\|-1)$ is non-increasing as $h\to 0^+$.
 \end{lemma}
 
\begin{proof}
Let $\alpha<1$. Since $\|v\|=1$,
\[
\begin{array}{rcl}
\varphi(v,\alpha h)&=&\displaystyle\frac{1}{\alpha h}\left(\|v+\alpha hAv\|-1\right)\\
&=&\displaystyle\frac{1}{h}\left(\|\frac{v}{\alpha}+hAv\|-\frac{1}{\alpha}\right)\\
&=&\displaystyle\frac{1}{h}\left(\|(\frac{v}{\alpha}-v)+(v+hAv)\|-(\frac{1}{\alpha}-1)-1\right)\\
&\leq&\displaystyle\frac{1}{h}\left(\|(\frac{v}{\alpha}-v)\|+\|(v+hAv)\|-(\frac{1}{\alpha}-1)-1\right)\\
&=&\displaystyle\frac{1}{h}(\|v+hAv\|-1)\\
&=&\varphi(v,h).
 \end{array}
\]
\end{proof}

\begin{corollary}\label{n-dec}
For any matrix $A$, $\displaystyle\frac{1}{h}(\|I+hA\|-1)$ is non-increasing as $h\to 0^+$.
 \end{corollary}
 
\begin{proof}
Let $\alpha<1$. By the definition of norm of a matrix and Lemma \ref{f-dec}
\[
\begin{array}{rcl}
\displaystyle\frac{1}{\alpha h}(\|I+\alpha hA\|-1)&=&\displaystyle\sup_{\|v\|=1}\frac{1}{\alpha h}\left(\|v+\alpha hAv\|-1\right)\\
&=&\displaystyle\sup_{\|v\|=1}\varphi(v,\alpha h)\\
&\leq&\displaystyle\sup_{\|v\|=1}\varphi(v, h)\\
&=&\displaystyle\frac{1}{h}(\|I+hA\|-1).
 \end{array}
\]
\end{proof}

\begin{proofoftheoremmu} 
{\em Claim $1$.} $$\displaystyle\sup_{v\in S^{nN-1}}\inf_{h\in (0,1)}\varphi(v,h)=\displaystyle\inf_{h\in (0,1)}\sup_{v\in S^{nN-1}}\varphi(v,h).$$

{\em Proof of Claim $1$.} To apply Proposition \ref{Joo}, we'll show that for $c<c^*$, $\displaystyle\bigcap_{h\in (0,1)}H_{h,c}\neq\emptyset$, where $H_{h,c}=\{v\in S^{nN-1}: \varphi(v,h)\geq c\}$ and $c^*$ is defined as above. By Lemma \ref{f-dec}, $\varphi(v,h)$ is decreasing in $h$ which implies $ H_{h_1,c}\subset H_{h_2,c}$ when $h_1<h_2$. Also by the definition of $c^*$, $c<c^*$ implies that $H_{h,c}\neq\emptyset$ for any $h\in (0,1)$.  On the other hand, each $H_{h,c}$ is a closed subset of $S^{nN-1}$, so they are all compact. Hence their intersection is non-empty.

{\em Claim $2$.} $$\sup_{v\in S^{nN-1}}\lim_{h\to0^+}\varphi(v,h)=\sup_{v\in S^{nN-1}}\inf_{h\in (0,1)}\varphi(v,h)$$ and 
$$\lim_{h\to0^+}\sup_{v\in S^{nN-1}}\varphi(v,h)=\inf_{h\in (0,1)}\sup_{v\in S^{nN-1}}\varphi(v,h).$$

{\em Proof of claim $2$.} By Lemma \ref{f-dec}, since $f(v,h)$ is non-increasing as $h\to0^+$, $\displaystyle\sup_{v\in S^{nN-1}}\lim_{h\to0^+}\varphi(v,h)=\displaystyle\sup_{v\in S^{nN-1}}\inf_{h\in (0,1)}\varphi(v,h)$. By Corollary \ref{n-dec}, since $\displaystyle\frac{1}{h}(\|I+hA\|-1)$ is non-increasing as $h\to0^+$, 
$\displaystyle\lim_{h\to0^+}\sup_{v\in S^{nN-1}}\varphi(v,h)=\displaystyle\inf_{h\in (0,1)}\sup_{v\in S^{nN-1}}\varphi(v,h).$
By claim $1$, the right hand side of the equalities in Claim $2$ are equal, and therefore so are their left hand sides:
\[\displaystyle\sup_{v\in S^{nN-1}}\lim_{h\to0^+}\varphi(v,h)=\displaystyle\lim_{h\to0^+}\sup_{v\in S^{nN-1}}\varphi(v,h),\]
 which implies $\mu(A)=\displaystyle\sup_{\|v\|=1}\lim_{h\to 0^+}\frac{1}{h}(\|v+hAv\|-1)$.\qed
\end{proofoftheoremmu}

\subsection{Proof of Lemma \ref{p-limit}.}
 \begin{proof}
Fix $p<q$. Then there exists $r>0$ such that $\displaystyle\frac{1}{q}+\frac{1}{r}=\displaystyle\frac{1}{p}$. Indeed $r=\displaystyle\frac{1}{\frac{1}{p}-\frac{1}{q}}$. Using Holder's inequality, 
\[
\begin{array}{rcl}
\displaystyle\left(\int_{\Omega}\abs{f}^p\right)^{\frac{1}{p}}&\leq&\displaystyle\left(\int_{\Omega}\abs{f}^q\right)^{\frac{1}{q}}\left(\int_{\Omega}1^r\right)^{\frac{1}{r}}\\
&\leq&\displaystyle\left(\int_{\Omega}\abs{f}^q\right)^{\frac{1}{q}}\abs{\Omega}^{\frac{1}{p}-\frac{1}{q}}\\
\Rightarrow\;\; \displaystyle\left(\frac{1}{\abs{\Omega}}\int_{\Omega}\abs{f}^p\right)^{\frac{1}{p}}&\leq&\displaystyle\left(\frac{1}{\abs{\Omega}}\int_{\Omega}\abs{f}^q\right)^{\frac{1}{q}}.
\end{array}
\]
Hence $F$ is an increasing function of $p$. Now we'll show that as $p\to\infty$, $F(p)\to\|f\|_{\infty}$. Note that for any $p$, $F(p)\leq\|f\|_{\infty}$ and therefore 
$\displaystyle\lim_{p\to\infty}F(p)\leq\|f\|_{\infty}$. To prove the converse inequality, for any $\epsilon>0$, we define $E_{\epsilon}:=\displaystyle\{\omega\in\Omega: \abs{f(\omega)}>\|f\|_{\infty}-\epsilon\}$ which by the definition of $\|f\|_{\infty}$ has positive measure, i.e. $\displaystyle\abs{E_{\epsilon}}>0.$
\[
\begin{array}{rcl}
\displaystyle\left(\int_{\Omega}\abs{f}^p\right)^{\frac{1}{p}}&\geq&\displaystyle\left(\int_{E_{\epsilon}}\abs{f}^p\right)^{\frac{1}{p}}\\
&\geq&\displaystyle(\|f\|_{\infty}-\epsilon)\left(\frac{\abs{E_{\epsilon}}}{\abs{\Omega}}\right)^{\frac{1}{p}}
\end{array}
\]
But because $\displaystyle\frac{\abs{E_{\epsilon}}}{\abs{\Omega}}>0$, $\displaystyle\left(\frac{\abs{E_{\epsilon}}}{\abs{\Omega}}\right)^{\frac{1}{p}}\to1$ as $p\to\infty$. Therefore for any arbitrary $\epsilon>0$, $$\displaystyle\lim_{p\to\infty}F(p)\geq\|f\|_{\infty}-\epsilon,$$ which implies $$\displaystyle\lim_{p\to\infty}F(p)\geq\|f\|_{\infty}.$$
\end{proof}


\end{document}